\documentclass[oneside,11pt]{article}
\usepackage{graphicx, subfigure}
\usepackage[top=1in,bottom=1in,left=1in,right=1in]{geometry}
\usepackage{amsfonts, amsmath, amssymb, amsthm, bbm, booktabs}
\usepackage{constants}
\usepackage{enumitem}
\usepackage{multirow}
\usepackage{hyperref}
\usepackage[mathscr]{euscript}
\usepackage{pifont}
\usepackage{mathtools}
\usepackage{natbib}
\usepackage{prodint}
\usepackage[normalem]{ulem}
\usepackage{comment}
\usepackage{etoc} 
\usepackage{mathrsfs}

\usepackage{color}
\definecolor{darkred}{RGB}{100,0,0}
\definecolor{darkgreen}{RGB}{0,100,0}
\definecolor{darkblue}{RGB}{0,0,150}

\usepackage{hyperref}
\hypersetup{colorlinks=true, linkcolor=darkred, citecolor=darkgreen, urlcolor=darkblue}
\usepackage{url}
\usepackage[dvipsnames]{xcolor}

\usepackage{cancel}
\usepackage{xcolor}

\usepackage{amsfonts, amsmath, amssymb, amsthm, bbm}
\usepackage{constants}
\usepackage{enumitem}
\usepackage{multirow}
\usepackage[mathscr]{euscript}
\usepackage{pifont}
\usepackage{mathtools}
\usepackage{natbib}
\usepackage{prodint}
\usepackage{mathrsfs}

\newtheorem{thm}{Theorem}
\newtheorem*{thm*}{Theorem}

\newtheorem{lem}{Lemma}
\newtheorem{cor}{Corollary}
\newtheorem{assump}{Assumption}
\newtheorem*{def*}{Definition}

\theoremstyle{remark}

\newtheorem{rem}{Remark}

\def\beqn{\begin{eqnarray*}}
\def\eeqn{\end{eqnarray*}}
\def\Bitem{\begin{itemize}\setlength{\itemsep}{.2in}}
\def\bitem{\begin{itemize}\setlength{\itemsep}{.05in}}
\def\eitem{\end{itemize}}
\def\Benum{\begin{enumerate}\setlength{\itemsep}{.2in}}
\def\benum{\begin{enumerate}\setlength{\itemsep}{.05in}}
\def\eenum{\end{enumerate}}
\def\bmult{\begin{multline*}}
\def\emult{\end{multline*}}
\def\bcenter{\begin{center}}
\def\ecenter{\end{center}}
\def\bframe{\begin{frame}}
\def\eframe{\end{frame}}

\newcommand{\thmref}[1]{Theorem~\ref{thm:#1}}
\newcommand{\assumpref}[1]{Assumption~\ref{assump:#1}}

\newcommand{\corref}[1]{Corollary~\ref{cor:#1}}
\newcommand{\lemref}[1]{Lemma~\ref{lem:#1}}
\newcommand{\secref}[1]{Section~\ref{sec:#1}}
\newcommand{\figref}[1]{Figure~\ref{fig:#1}}

\def\cD{\mathcal{D}}

\def\cF{\mathcal{F}}

\def\cI{\mathcal{I}}

\def\cL{\mathcal{L}}

\def\cN{\mathcal{N}}

\def\cX{\mathcal{X}}

\def\bbR{\mathbb{R}}

\renewcommand{\P}{\operatorname{\mathcal{P}}}
\newcommand{\p}{\operatorname{\mathrm{pr}}}
\newcommand{\var}{\operatorname{\mathrm{var}}}

\def\Unif{\text{Unif}}

\def\1{\mathbbm{1}}

\def\TV{\mathrm{TV}}
\def\f{f_\gamma}
\def\d{\frac{d}{d \gamma}}
\def\opr{o_{p}(n^{-1/2})}
\def\cf{\mathrm{cf}}
\newcommand{\abs}[1]{\left| #1 \right|}
\newcommand{\cur}[1]{\left\{ #1 \right\}}
\newcommand{\squ}[1]{\left[ #1 \right]}
\newcommand{\bra}[1]{\left( #1 \right)}

\newcommand{\nsuptheta}[1]{\left\| #1 \right\|_{\Theta}}

\def\supt{\sup_{t \in [0,\tau]}}

\def\supth{\sup_{\theta \in \cD}}
\def\suptl{\sup_{t,l}}
\def\suptheta{\sup_{\theta \in \Theta}}

\def\suplam{\sup_{t\in[0,\tau]} \left|\hat\Lambda(t|L)-\Lambda_0(t|L)\right|}
\def\supmu{\sup_{t\in[0,\tau]} \left|\hat\mu(t,L)-\mu_0(t,L)\right|}
\def\suph{\sup_{t\in[0,\tau]} |h(t,L)|}
\def\suplaml{\sup_{t\in[0,\tau]} \left|\hat\Lambda(t|l)-\Lambda_0(t|l)\right|}

\def\suplami{\sup_{t\in[0,\tau]} \left|\hat\Lambda(t|L_i)-\Lambda_0(t|L_i)\right|}
\def\supmui{\sup_{t\in[0,\tau]} \left|\hat\mu(t,L_i)-\mu_0(t,L_i)\right|}

\def\suplamk{\sup_{t\in[0,\tau]} \left|\hat\Lambda_{-k}(t|L)-\Lambda_0(t|L)\right|}
\def\supmuk{\sup_{t\in[0,\tau]} \left|\hat\mu_{-k}(t,L)-\mu_0(t,L)\right|}

\def\nsuplam{\left\| \hat \Lambda - \Lambda_0 \right\|_{\sup,2}}
\def\nsupmu{\left\| \hat \mu - \mu_0 \right\|_{\sup,2}}

\def\edg{E_\dagger}

\def\dgsuplam{\sup_{t\in[0,\tau]} |\hat\Lambda(t|L_\dagger)-\Lambda_0(t|L_\dagger)|}

\def\dgsupmu{\sup_{t\in[0,\tau]} |\hat\mu(t,L_\dagger)-\mu_0(t,L_\dagger)|}

\def\dgsupnlam{\| \hat \Lambda - \Lambda_0 \|_{\dagger,\sup,2}}

\def\dgsupnmu{\| \hat \mu - \mu_0 \|_{\dagger,\sup,2}}

\def\dgsupnmuk{\| \hat \mu_{-k} - \mu_0 \|_{\dagger,\sup,2}}
\def\dgsupnlamk{\| \hat \Lambda_{-k} - \Lambda_0 \|_{\dagger,\sup,2}}

\newcommand{\suppformatappendix}{
    \setcounter{section}{0}
    \setcounter{subsection}{0}
    \setcounter{equation}{0}
    \setcounter{figure}{0}
    \setcounter{table}{0}
    
    \renewcommand{\thesection}{S\arabic{section}}  
    \renewcommand{\thesubsection}{\thesection.\arabic{subsection}} 
    \renewcommand{\theequation}{S\arabic{equation}} 
    \renewcommand{\thefigure}{S\arabic{figure}}
    \renewcommand{\thetable}{S\arabic{table}}

    \renewcommand{\theHsection}{appendix.section.\thesection}
    \renewcommand{\theHsubsection}{appendix.subsection.\thesubsection}
    \renewcommand{\theHfigure}{appendix.figure.\thefigure}
    \renewcommand{\theHtable}{appendix.table.\thetable}
    \renewcommand{\theHequation}{appendix.equation.\theequation}
}

\title{\bf{Efficient Inference for Incremental Causal Effects of Time to Treatment}}

\author{
Zhichen Zhao$^{\dagger}$ \and
Andrew Ying \and
Ronghui Xu$^{\dagger,\ddagger,\S}$
}

\date{\small
$^{\dagger}$Department of Mathematics \\
$^{\ddagger}$Herbert Wertheim School of Public Health \& Human Longevity Science \\
$^{\S}$Hal{\i}c{\i}o{\u{g}}lu Data Science Institute \\
University of California San Diego
}

\begin{document}


\maketitle

\begin{abstract}
We consider continuous time to treatment initiation. This can commonly occur in preventive medicine, such as disease screening and vaccination; it can also occur with non-fatal health conditions such as HIV infection without the onset of AIDS. While traditional causal inference focused on `when to treat' and its effects, 
we consider the incremental causal effect when the intensity of time to treatment initiation is intervened upon. We derive the efficient influence function for this estimand and develop an estimation framework that accommodates flexible machine learning methods while achieving fast convergence rates. 
Valid confidence bands are obtained leveraging empirical process theory. 
We illustrate our approach via simulation, 
and apply it to cervical cancer screening data to study the incremental effect of time to subsequent HPV testing on cervical intraepithelial neoplasia detection.
\end{abstract}


\section{Introduction} 

Causal inference studies “what happens” under different types of interventions. Early work focused on \emph{static deterministic} interventions \citep{rubin1974estimating, hernan2020causal}, where treatment is fixed uniformly across individuals; for example, the average treatment effect (ATE) compares outcomes under regimes in which everyone is treated versus everyone is untreated. Here, “static” contrasts with “dynamic” interventions, which allow treatment to depend on covariates and evolving histories \citep{robins1986new, murphy2003optimal, robins2004optimal, moodie2007demystifying, young2011comparative, haneuse2013estimation, rytgaard2022continuous, ying2024causality2}, and “deterministic” contrasts with “stochastic” interventions, which assign treatment according to a probability distribution \citep{cain2010start, diaz2012population, van2018stochastic, diaz2020causal}.

Stochastic dynamic interventions have received more attention in the literature recently.  
A stochastic dynamic intervention, which shift the odds of receiving a binary treatment, has been referred to as an \emph{incremental intervention} \citep{kennedy2019nonparametric, naimi2021incremental, kim2021incremental, sarvet2023longitudinal, bonvini2023incremental}. 
This can be seen as `interpolates' between the more extreme scenarios of ``everyone is treated'' and ``everyone is untreated.''
The corresponding causal estimand under an incremental intervention is called the incremental causal effect.  
Unlike the ATE which helps to answer questions for an individual subject like ``what happens if I get the treatment versus not getting the treatment,'' the incremental causal effect 
 can answer policy related question such as, if we provide incentive or improve accessibility to health services, how that might affect the outcomes in a certain population \citep{bonvini2023incremental}.

While the literature has begun to explore these incremental effects in discrete-time settings \citep{kennedy2019nonparametric, sarvet2023longitudinal}, many processes of interest occur in continuous time.
\cite{ying2025incremental}. first generalized the incremental intervention framework to allow for continuous time to treatment initialization, established its identification 
and proposed an estimation framework using inverse probability weighting (IPW).

As a motivating example, consider cervical cancer screening in Norway \citep{roysland2025graphical}, where 
  women aged 25–69 have been advised since 1995 to have the cytology exam every three years. The goal of the exam is to identify  those with cervical intraepithelial lesion grade 2 or 3 (CIN2+) which can then be treated. Some of the cytology exams yield inconclusive results, and since 2005 human papillomavirus (HPV)  testing has been used to guide future treatment strategies. 
  Between 2005 and 2010, the three most commonly used HPV tests in Norway were Amplicor, Hybrid Capture 2 (HC2), and PreTectProofer \citep{nygard2014comparative}. 
When used following an inconclusive cytology finding, negative results from the PreTectProofer test have been observed to be more frequently followed by eventual detection of CIN2+ compared to its competitors, suggesting a higher false-negative rate \citep{haldorsen2011secondaryscreening, nygard2014comparative, roysland2025graphical}. However, individuals in the PreTectProofer group were also subject to more  subsequent testing \citep{nygard2014comparative}, as shown in \figref{figures/hpv_KM}. 
\begin{figure}[h]
\centering
\includegraphics[width=0.5 \textwidth]{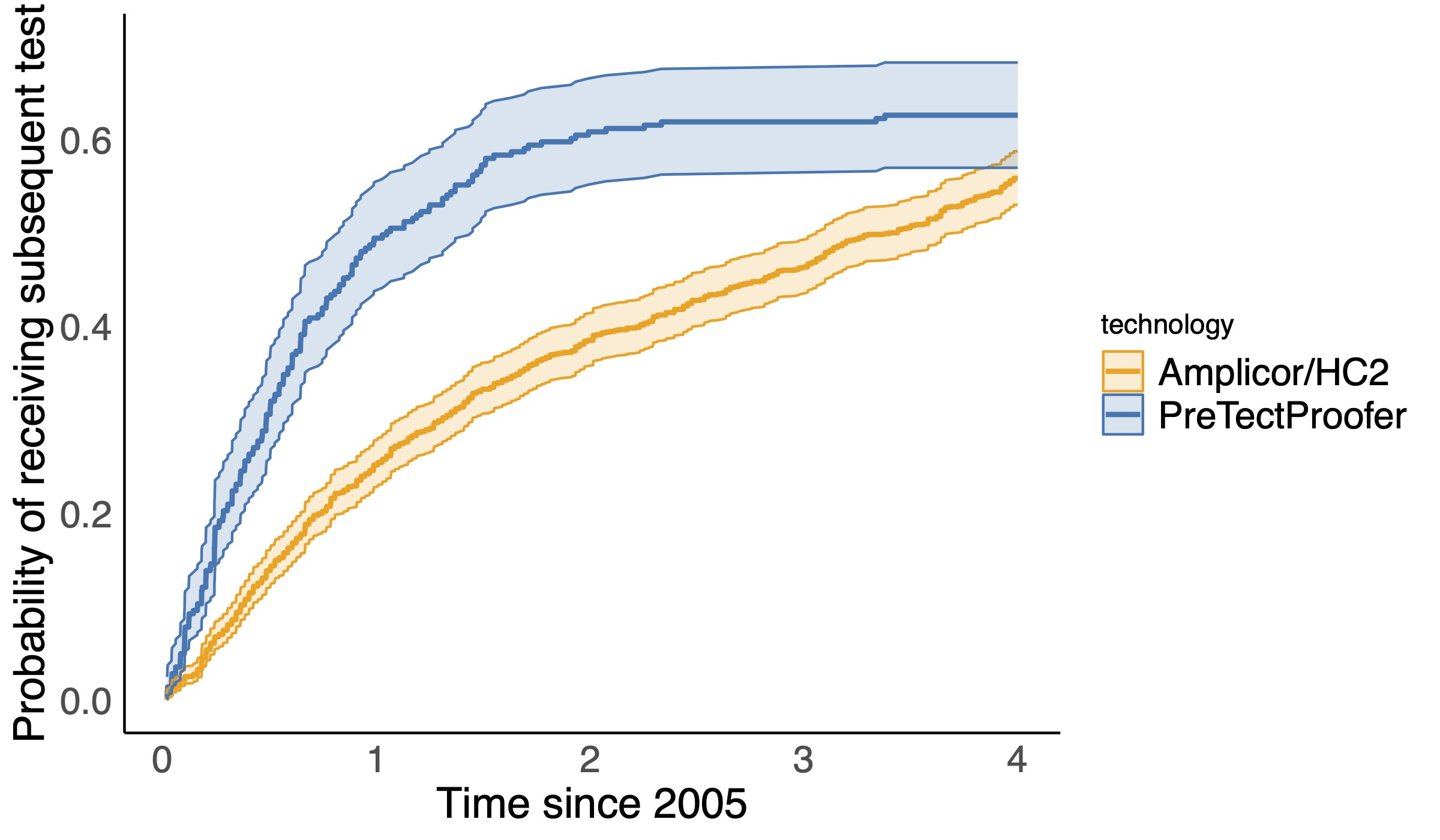}
\caption{Kaplan-Meier estimate of probability of receiving subsequent testing 
with pointwise 95\% confidence intervals.}
\label{fig:figures/hpv_KM}
\end{figure}
A natural question is whether the higher detection rate is due to more frequent testing, and what would the CIN2+ detection rates look like if women from the PreTectProofer group were subject to the same frequency of follow-up testing as the other groups?

In addition to its apparent different interpretation, 
another feature of  the incremental causal effect is that it does not require positivity for identification from the observed data distribution. 
Positivity is necessary in order to identify 
traditional causal estimands such as the ATE, which requires that every individual have a positive probability of receiving each treatment level. 
In practice positivity can be unrealistic, in particular with multiple time points or in continuous time. 
Incremental  interventions, on the other hand, does not assign treatment to different levels but instead perturbs the  treatment distribution. As a result     
{positivity} assumption is not needed in order to identify the incremental causal effects 
\citep{kennedy2019nonparametric}.

While \cite{ying2025incremental} developed the IPW approach, it is not optimal in the senses that for the nuisance parameters that are involved, it requires correct (semi-)parametric model specification in order to guarantee consistency of estimation, and 
 modern flexible nonparametric or machine learning (ML) methods cannot be applied for nuisance estimation while still having $\sqrt{n}$-inference. 
In contrast, the efficient influence function which possesses Neyman orthogonality \citep{tsiatis2006semiparametric, chernozhukov2018double}, can provide $\sqrt{n}$-consistent and asymptotically normal estimators under weaker conditions on the nuisance estimation \citep{chernozhukov2018double, kennedy2019nonparametric}.
In this paper, building on \cite{ying2025incremental}, we derive the efficient influence function for the incremental effect of time to treatment. We study the resulting estimators that allow both (semi)parametric modeling and flexible nonparametric and machine learning approaches for nuisance estimation. We further develop asymptotically valid uniform confidence bands under any arbitrary classes of incremental interventions, with minor conditions using empirical process theory. The estimators are applied to the Norwagian cervical cancer screening data introduced above. 

\section{
Incremental intervention and effect} \label{sec:iden}


Let $T$ be time to (initializing) a certain treatment, and $Y$ is an outcome of interest measured at time $\tau$. 
Here the treatment `process' is a 0-1 one: an individual is not treated until time $T$, when they become treated. In discrete time this is often referred to as event study. 
We observe $Y$, $ U = T \wedge \tau$, and baseline covariates $L$. 
Let $\Delta  = 1(T < \tau) = 1(U < \tau)$. 
We assume that $T$ is absolutely continuous with support $[0, \infty)$,  and $L$ has support $\cL$. The conditional hazard function of $T$ at time $t$ given $L = l$ is $\lambda(t|l) = \lim_{h \to 0}\p(t \leq T < t + h| T \geq t, L = l)/h$, and 
$\Lambda(t|l) = \int_0^t\lambda(s|l)ds$.
Note that the hazard function is a natural quantity  to describe the intensity of time to treatment \citep{yang2018modeling, yang2020semiparametric}.

Before discussing the incremental intervention and effect under the above time to treatment setting, as contrast we briefly review  the methodology  and common applications of deterministic (i.e.~non-stochastic) intervention and their causal effects. Later in this section we also contrast the identification of the causal effects under these two different types of interventions, in particular with regard to the positivity assumption. 
The deterministic  intervention for time to treatment  falls under the general dynamic treatment regime pioneered in \cite{murphy2001marginal},  
and was considered under the marginal structure model setup for continuous time to treatment using semiparametric Cox regression type modeling in  \cite{johnson2005semiparametric} and \cite{yang2018modeling}. Meanwhile matching approaches under this setting have also been considered in \cite{li:rosen} and \cite{lu:2005}. In addition, \cite{yang2020semiparametric} considered structural nested accelerated failure time model, and \cite{picciotto2012structural} considered structural nested cumulative failure time models with the concept of a `blip' function. Finally optimal treatment regimen under this setting was considered in \cite{nie2021learning}. 
All these deterministic intervention methodologies are applicable to time to treatment initialization or discontinuation, and in our own work on  vaccination during pregnancy \citep{cham:2013, xu:luo:glyn, cham:2016}. They answer the types of questions for an individual patient or pregnant woman about ever treated versus never treated, or about the effect of timing of the treatment. 

In this paper 
as motivated by the HPV testing example in the Introduction, 
 we would like to answer a different type of health \emph{policy relevant} questions like,  ``what if the subsequent testing intensity of women from the PreTectProofer group were reduce to that of the Amplicor/HC2 groups?''
In general, we may consider modifying the hazard $\lambda(t|l)$  by any multiplicative factor $\theta(t, l) > 0$. 
Note that a constant $\theta$ may give the simplest interpretation,  and can be termed a \emph{proportional hazards} intervention.
More generally though, we would want to allow, from a policy point of view, recommendation of intensified disease screening for certain high risk groups described via $l$, or after a certain age $t$. 

We define the potential outcome \citep{neyman1923applications, rubin1974estimating, holland1986statistics} $Y_t$ if the treatment occurs at time $t$. 
Since the outcome is measured at time $\tau$, for $t \geq \tau$ we denote $Y_t = Y_{\tau} = Y_\infty$ for someone who has not been treated by the time $\tau$. 
A static deterministic intervention with a target estimand like the ATE would contrast $E(Y_t)$ versus $E(Y_s)$ for some $t\neq s$. On the other hand, 
in this paper we consider $T(\theta)$  a random draw from the distribution with hazard function $\theta(t,l)\cdot\lambda(t|l)$. 
The corresponding potential outcome under $T(\theta)$ is then 
$Y_{T(\theta)}$, and the incremental causal effect index by $\theta$ is defined as
\begin{equation}\label{eq:para}
    \psi(\theta) = E\{Y_{T(\theta)}\}.
\end{equation}
In particular when $\theta\equiv 1$, $\psi(1) = E\{Y_{T(1)}\} = E(Y)$ corresponds to the expectation under the factual distribution of $T$ that we have observed.

\cite{ying2025incremental} showed that under Assumptions \ref{assump:cons} and \ref{assump:sra} below, 
the incremental causal effect curve $\psi(\theta)$ can be identified by the observed data distribution via
\begin{equation}\label{eq:iden}
    \psi(\theta) = E\{Y_{T(\theta)}\} = E\left[ Y \theta(T,L)^{\Delta}e^{-\int_0^{T \wedge \tau}\{\theta(t,L) - 1\}d\Lambda(t|L)} \right].
\end{equation}
\begin{assump}[Consistency] \label{assump:cons}
$    Y_{T \wedge \tau} = Y$.
\end{assump}
\begin{assump}[No Unmeasured Confounding]\label{assump:sra}
$  T \perp Y_t~|~L$, for all $t \in [0, \tau] $.
\end{assump}

\begin{rem}
We note that Assumptions \ref{assump:cons} and \ref{assump:sra} are commonly assumed for causal inference. On the other hand, positivity is not required. This can be understood from the heuristics below, and again we contrast with deterministic interventions. 
The observed data likelihood for a single observation is
\begin{align*}
    P = \p(L) \lambda(T|L)^{\Delta} e^{-\int_0^{T \wedge \tau} d\Lambda(t|L)} \p(Y|T\wedge\tau, L).
\end{align*}
Traditional non-stochastic dynamic interventions 
typically operate by replacing 
$\lambda(T|L)^{\Delta} e^{-\int_0^{T \wedge \tau} d\Lambda(t|L)}$ in the above by some pre-specified conditional distribution of $T$ given $L$. This necessitates a positivity assumption: the support of the intervention distribution must be a subset of the support of the observed distribution, in order for the Radon-Nikodym derivative to exist. 
In contrast, the target likelihood under the incremental intervention is
\begin{align*}
    P_{\theta} = \p(L)\{\theta(T,L) \lambda(T|L)\}^{\Delta} e^{-\int_0^{T \wedge \tau} \theta(t,L) d\Lambda(t|L)} \p(Y|T\wedge\tau, L).
\end{align*}
Crucially, one can see that when $P = 0$, $P_\theta = 0$: hence $P_\theta$ is absolutely continuous with respect to $P$ without any positivity assumption. The Radon-Nikodym derivative that gives rise to \eqref{eq:iden} is
\begin{equation*} 
    \frac{\textup{d}P_{\theta}}{\textup{d}P} = \frac{\{\theta(T,L) \lambda(T|L)\}^{\Delta} e^{-\int_0^{T \wedge \tau} \theta(t,L) d\Lambda(t|L)}} 
    {\lambda(T|L)^{\Delta} e^{-\int_0^{T \wedge \tau} d\Lambda(t|L)}} = {\theta(T,L)}^{\Delta} e^{-\int_0^{T \wedge \tau} \{\theta(t,L)-1\} d\Lambda(t|L)}.
\end{equation*}
\end{rem}

The identification in \eqref{eq:iden} naturally leads to an IPW estimator in  \cite{ying2025incremental}:
$$\hat \psi_{\text{ipw}}(\theta) = \frac{1}{n}\sum_{i = 1}^n Y_i \theta(T_i, L_i)^{\Delta_i}e^{-\int_0^{T_i \wedge \tau}\{\theta(t, L_i) - 1\}d\hat \Lambda(t|L_i)},
$$ 
where $\hat\Lambda(t|l)$ is  estimated from data.
In order to guarantee $\sqrt{n}$-inference, (semi-)parametric models have to be used for $\hat\Lambda(t|l)$, which is subject to model misspecification. 
On the other hand, efficient influence function (EIF) possesses the property of \textit{Neyman orthogonality} \citep{tsiatis2006semiparametric, chernozhukov2018double}, and as a result flexible nonparametric or machine learning estimators can be used for nuisance estimation while maintaining $\sqrt{n}$-inference for $\psi(\theta)$ with cross-fitting. This is what we will pursue below.

\section{Efficient influence function}


As discussed above, the efficient influence function plays a central role in constructing estimators that admit valid $\sqrt{n}$-inference while allowing the nuisance functions to be estimated by machine learning (ML) methods. It is a core concept in semiparametric theory \citep{bickel1993efficient, van2000asymptotic, tsiatis2006semiparametric, kosorok2008introduction, kennedy2016semiparametric, ying2026deepening}, which studies efficient estimation of a target parameter in infinite-dimensional models, including the nonparametric model considered here. Intuitively, the influence function identifies the component of variation in the observed data distribution that is most directly informative about the target parameter, while being orthogonal to directions that primarily change nuisance features of the model. This is closely related to the Neyman orthogonality property discussed earlier: the resulting estimating function inspired by the influence function, is locally insensitive, to the first order, to errors in nuisance estimation. Consequently, estimation errors in the nuisance functions can affect the target parameter estimator only through second-order remainder terms. This orthogonality is what permits a broad class of ML methods to be used for nuisance estimation without invalidating $\sqrt{n}$-inference, provided the nuisance estimators converge sufficiently fast.
Under a nonparametric model such as ours, the influence function of the estimand is unique and therefore coincides with the efficient influence function, which attains the semiparametric efficiency bound. 
 \citep{tsiatis2006semiparametric}. 
In this section we derive the EIF for the incremental causal effect  $\psi(\theta)$ following the approach in \cite{hines2022demystifying}.

Denote $\mu(t, l) = E(Y|U = t, L = l)$ for $t \in [0, \tau]$ and $l \in \cL$. 
For a real-valued function $f(t,l)$ with $t \in [0, \tau]$ and $l \in \cL$, the total variation of $f(\cdot,l)$ on the interval $[0, \tau]$ is $\TV\{f(\cdot,l)\} = \sup_{\Pi}\sum_{j=0}^{J-1}|f(x_{j+1},l)-f(x_{j},l)|$, 
where $\Pi$ denotes the set of all possible partitions $0=x_{0}<x_{1}<...<x_{J}=\tau$ of $[0, \tau]$. 
In the following subscript `0' denotes true value of the parameter under which the data is generated. 

\begin{assump}[Boundedness] \label{assump:bound}
Assume:
\begin{enumerate}[label=\theassump\alph*., ref=\theassump\alph*]
\item \label{assump:bound:y}
$E(Y^4) < \infty$. 
\item \label{assump:bound:theta}
There exist constants $0 < c \leq C < \infty$ such that $c \leq \theta(t,l) \leq C$, for all $(t,l) \in [0,\tau] \times \cL$.
\item \label{assump:bound:nuisance}
The functions $\Lambda_0(t|l)$ and $\mu_0(t, l)$ are uniformly bounded over $t \in [0,\tau]$ and $l \in \cL$.
\item \label{assump:bound:tv}
$\TV\cur{\theta(\cdot, l)}$ and $\TV\cur{\mu_0(\cdot,l)}$ are bounded over $l \in \cL$.
\end{enumerate}
\end{assump}
Assumptions~\ref{assump:bound:y} - \ref{assump:bound:nuisance} are  regularity conditions to ensure that the EIF in \thmref{eifiden} below has finite variance.
\assumpref{bound:tv} 
 is a mild requirement that the functions are not too `jagged' over time $t$, and 
 satisfied by a broad class including Lipschitz and absolutely continuous functions.

\begin{thm}[Efficient Influence Function]\label{thm:eifiden}
Under Assumptions \ref{assump:cons} - \ref{assump:bound}, the efficient influence function for the incremental causal effect $\psi(\theta)$ is given by $\phi(\theta; \Lambda_0, \mu_0) -\psi(\theta)$, where
\begin{align}
\phi(\theta; \Lambda_0, \mu_0) =& Y \theta(U,L)^{\Delta} e^{-\int_0^U\{\theta(v,L) - 1\}d\Lambda_0(v|L)} \notag \\
& + \int_0^\tau \mu_0(u, L) \theta(u, L)^\delta e^{-\int_0^u\{\theta(v, L)-1\} d \Lambda_0(v | L)} \left[\int_0^{U \wedge u} \left\{\theta(v, L)-1\right\} d e^{\Lambda_0(v | L)} \right] d F_0(u|L) \notag \\
& - \frac{\theta(U, L) - 1}{e^{-\Lambda_0(U | L)}} \int_{U+}^\tau \mu_0(u, L) \theta(u, L)^\delta e^{-\int_0^u \{\theta(v, L)-1\} d \Lambda_0(v | L)} d F_0(u|L), \label{eq:org_eif}
\end{align}
 $\delta = 1(u < \tau)$ and $ F_0(u|l) = 1 -e^{- \Lambda_0(u|l)}$.
\end{thm}
The proof of \thmref{eifiden} is given in \ref{sec:pfeifiden} of the Supplementary Material. We can see that the EIF involves two nuisance parameters: $ \mu$ is the outcome regression, and $ \Lambda$ the treatment mechanism. 
The proof of \corref{eifiden} below is given in \secref{pfeifidencor} of the Supplementary Material.

\begin{cor}\label{cor:eifiden}
When 
$\theta(t, l) \equiv \theta\in \bbR$, the EIF for $\psi(\theta)$ in \thmref{eifiden} simplifies to:
\begin{align}\label{eq:eifconst}
& Y \theta^\Delta e^{-(\theta - 1) \Lambda_0(U | L)} 
 + (\theta - 1) \int_0^U \mu_0(u, L) \theta^\delta e^{-(\theta-2)\Lambda_0(u|L)} d F_0(u|L) \notag \\
& - (\theta - 1) \int_0^\tau \mu_0(u, L) \theta^\delta e^{-(\theta-1)\Lambda_0(u|L)} d F_0(u|L)
- \psi(\theta).
\end{align}
\end{cor}

\section{Estimation and Inference} \label{sec:est}

The efficient influence function  can be used to construct estimators with favorable asymptotic properties, such as fast convergence rates even when nuisance functions are estimated nonparametrically. 
Considering  a random sample $\{Y_i, U_i, L_i, \Delta_i\}_{i=1}^n$ from the distribution of $(Y, U, L, \Delta)$. 
According to \thmref{eifiden}, we have $\psi(\theta) = E \{\phi(\theta; \Lambda_0, \mu_0)\}$.
In the Supplementary Material (\secref{eifeq}) we show 
a slightly different algebraic form of \eqref{eq:org_eif}, which leads immediately to
the augmented inverse probability weighted (AIPW) estimator:
\begin{align} 
\hat \psi(\theta) =& \frac{1}{n}\sum_{i=1}^n \phi_i(\theta; \hat \Lambda, \hat \mu) \notag \\
=& \frac{1}{n} \sum_{i=1}^n \bigg[ Y_i \theta(U_i,L_i)^{\Delta_i} e^{-\int_0^{U_i}\{\theta(v,L_i) - 1\}d \hat \Lambda(v|L_i)} \notag \\
& \phantom{\frac{1}{n} \sum_{i=1}^n \Big[ } - \int_0^\tau \hat \mu (u, L_i) \left[\int_0^{U_i \wedge u} \{\theta(v, L_i)-1\} d e^{\hat \Lambda(v | L_i)} \right] d e^{-\int_0^u \theta(v, L_i) d \hat \Lambda(v | L_i)} \notag \\
& \phantom{\frac{1}{n} \sum_{i=1}^n \Big[ } + \frac{\theta(U_i, L_i) - 1}{e^{-\hat \Lambda(U_i | L_i)}} \int_{U_i+}^\tau \hat \mu (u, L_i) d e^{-\int_0^u \theta(v, L_i) d \hat \Lambda(v | L_i)} \bigg], \label{eq:hat_psi}
\end{align}
where $\hat \Lambda(u|l)$ and $\hat \mu(u,l)$ are estimators of the nuisance functions $\Lambda(u|l)$ and $\mu(u,l)$, respectively. 

In practice, the nuisance functions can be estimated using either parametric or semiparametric models, or flexible nonparametric machine learning methods. 
In the following we first show that as long as the nuisance estimators are uniformly consistent, $\hat \psi(\theta)$ defined in \eqref{eq:hat_psi} is consistent for our estimand $\psi(\theta)$.

Define
\begin{equation*}
\nsuplam^2 = E \cur{\suplam^2}, \ \ 
\nsupmu^2 = E \cur{\supmu^2}.
\end{equation*}
\begin{assump} \label{assump:nuiest} 
Assume:
\begin{enumerate}[label=\theassump\alph*., ref=\theassump\alph*]
    \item \label{assump:nuiest:bound}
    $\hat\Lambda(\tau|L)$ is bounded almost surely.
    
    \item \label{assump:nuiest:cons}
    $\nsuplam = o(1)$ and $\nsupmu = o(1)$.
\end{enumerate}
\end{assump}
\assumpref{nuiest:bound} is typically satisfied by any common estimate of the cumulative hazard function, at a finite time $\tau$. 
Uniform consistency in \assumpref{nuiest:cons}, implied by weak convergence to a zero-mean Gaussian process, is known to be satisfied by the nonparametric Nelson-Aalen 
estimator, (semi-)parametric regression models including the commonly used Cox regression  \citep{tsiatis1981large, andersen1982cox}, and generalized linear models on compact regressor supports \citep{fahrmeir1985consistency}. 
Similarly, nonparametric learners like random forests \citep{scornet2015consistency, athey2019generalized} achieve uniform consistency provided that their complexity is properly regularized relative to the sample size. These conditions  fail in the presence of severe model misspecification or extreme high-dimensionality lacking sparsity \citep{buhlmann2011statistics, wainwright2019high}.

\begin{thm}[Consistency]\label{thm:eifcons}
Under Assumptions \ref{assump:cons} - \ref{assump:nuiest}, $\hat \psi(\theta)$ converges to $\psi(\theta)$ in probability as $n \to \infty$.
\end{thm}

The proof of \thmref{eifcons} is given in \secref{pfeifcons} of the Supplementary Material.

For inference purposes 
 we can show that when the nuisance estimators are regular and asymptotically linear (RAL), which is typically the case when parametric or semiparametric models are used,  the proposed estimator 
$\hat\psi(\theta)$ in \eqref{eq:hat_psi} is asymptotically normal. 
The proof is similar to that in \cite{wang2024doubly}, and is availabe from the authors upon request. The RAL nuisance estimators tend  not to require a very large sample size, and $\hat\psi(\theta)$ does not require cross-fitting. 

More generally, we consider flexible nonparametric ML methods to estimate the nuisance parameters in the following, since   
parametric or semiparametric  models may be misspecified  
in practice. 
As in the literature \citep{klaassen1987consistent, schick1986on}, 
 we employ a cross-fitting procedure where we randomly partition the data indices into $K$ folds of roughly equal size, denoted by $\cI_1,\dots,\cI_K$, and construct the estimator as
\begin{equation}
    \hat \psi_{\text{cf}}(\theta) = \frac{1}{n} \sum_{k=1}^K \sum_{i\in \cI_k} \phi_i(\theta; \hat \Lambda_{-k}, \hat \mu_{-k}), \label{eq:psi_cf}
\end{equation}
where the nuisance estimators $\hat \Lambda_{-k}$ and $\hat \mu_{-k}$ are estimated using the out-of-$k$-fold data indexed by $\cI_{-k} = \{1,\dots,n\}\setminus \cI_k$.

For the theory with ML nuisance estimation and cross fitting, 
we introduce some norms here. Denote by $O$ the data used to obtain $\hat \Lambda$ and $\hat \mu$. Let $O_\dagger = (Y_\dagger, U_\dagger, L_\dagger)$ denote a copy of the data that is independent of, but from the same underlying distribution, as $O$. Let $E_\dagger$ denote the expectation taken with respect to $O_\dagger$ conditional on the data $O$. Define
\begin{align*}
& \|\hat\Lambda-\Lambda_0\|_{\dagger,\sup,p}^p = \edg\cur{\dgsuplam^p}, \text{ for } p=2,4, \\
&\dgsupnmu^2 = \edg\cur{\dgsupmu^2}.
\end{align*}

In the following, we consider $\theta(t,l) \in \Theta \coloneqq\cur{\theta: [0,\tau]\times\cL \to \bbR}$. 

\begin{assump} \label{assump:cf}
Assume
\begin{enumerate}[label=\theassump\alph*., ref=\theassump\alph*]
    \item \label{assump:cf:sigma}
    $\suptheta 1/\sigma(\theta) \leq C_{\sigma} < \infty$.

    \item \label{assump:cf:donsker}
    The class $\cur{\phi(\theta;\Lambda_0,\mu_0): \theta \in \Theta}$ is Donsker.

    \item \label{assump:cf:prod_rate}
    $\dgsupnmu \cdot \dgsupnlam + \dgsupnlam^2 + \|\hat\Lambda-\Lambda_0\|_{\dagger,\sup,4}^2 = \opr.$
\end{enumerate}
\end{assump}

\assumpref{cf:sigma} requires the variance function of the EIF to be bounded away from zero. \assumpref{cf:donsker} is used to establish weak convergence. 
In the Supplementary Material (\secref{verify}), we verify that these two assumptions hold when $\theta(t, l) \equiv \theta \in \cD \coloneqq [c,C]$, 
where $0 < c \leq C < \infty$. 
The result can be extended to a more general class of functions $\theta(t,l;\gamma)$ parameterized by a vector of parameters $\gamma \in \Gamma \subseteq \bbR^d$, where $1 \leq d < \infty$ and $\Gamma$ is compact, provided that
$\theta(t,l;\gamma)$ is continuous in $\gamma$ and its partial derivative with respect to each component of $\gamma$ exists and is uniformly bounded over $(t,l)$. This would include piecewise constant or piecewise polynomial $ \theta (t,l)$, for example.

\assumpref{cf:prod_rate} gives the nuisance rate condition  for $\sqrt{n}$-inference of our estimand. It requires the nuisance estimators to converge at rate $o_p(n^{-1/4})$, as typically required by a Neyman orthogonal score including the EIF. 
Naturally it is achieved under correctly specified parametric or semiparametric models, therefore the simultaneous inference below applies to RAL nuisance estimators. 
When $Y$ is bounded, \assumpref{cf:prod_rate} can be relaxed to 
$ \dgsupnmu \cdot \dgsupnlam + \dgsupnlam^2 = o_p(n^{-1/2})$. 

We show that $\hat \psi_{\cf}(\theta)$ converges to a process over $\theta$. 
Let 
$\hat{\sigma}_{\cf}^2(\theta) =  \sum_{k=1}^K \sum_{i\in \cI_k} \{\phi_i(\theta; \hat \Lambda_{-k}, \hat \mu_{-k}) - \hat{\psi}_{\cf}(\theta)\}^2 /n $ be an estimator of the variance function $\sigma^2(\theta) = E[\{\phi(\theta; \Lambda_0, \mu_0) - \psi(\theta)\}^2]$. 

\begin{thm}[Process Convergence with Cross-Fitting] \label{thm:cf}
Under Assumptions \ref{assump:cons} - \ref{assump:cf}, 
\begin{equation}
    \frac{\hat{\psi}_{\cf}(\theta) - \psi(\theta)}{\hat{\sigma}_{\cf}(\theta)/\sqrt{n}} \rightsquigarrow G(\theta), \label{eq:process_con}
\end{equation}
in $\ell^\infty(\Theta)$, where $G(\cdot)$ is a mean-zero Gaussian process with covariance process
$$
E\{G(\theta_1)G(\theta_2)\} = E\{\tilde\phi(\theta_1; \Lambda_0, \mu_0) \tilde\phi(\theta_2; \Lambda_0, \mu_0)\},
$$
with $\tilde\phi(\theta; \Lambda_0, \mu_0) \equiv \{\phi(\theta; \Lambda_0, \mu_0) - \psi(\theta)\}/\sigma(\theta)$, and $\ell^\infty(\Theta)$ is the set of all uniformly bounded real functions on $\Theta$. 
\end{thm}

The proof of \thmref{cf} is given in \secref{cf_proof} of the Supplementary Material. 
It is inspired by the proof of Theorem 3 in \cite{kennedy2019nonparametric}, which uses bracketing numbers and 
properties of Donsker classes. In contrast to \cite{kennedy2019nonparametric}, however, 
part of the proof relies on a Gateaux derivative argument. This is because the EIF structure in the Kennedy paper is relatively simple, allowing the remainder term to be directly rewritten and bounded via products of nuisance estimation errors. 
On the other hand, the EIF here has a more complex integral form, 
and we have found the Gateaux derivative to be an effective way  to bound the relevant terms by products of estimation errors.

The result of \thmref{cf}  
enables uniform inference over a range of $\theta$ values. 
This can be useful because in practice as we are  likely to be interested in shifting the observed treatment distribution by a range of $\theta$ values, as opposed to a particular single value; this can be seen in our example later. 
In order to do so, we need to find a critical value $c_\alpha$ such that
\begin{equation}
    \p \left( \sup_{\theta \in \Theta} \left| \frac{\hat{\psi}_{\cf}(\theta) - \psi(\theta)}{\hat{\sigma}_{\cf}(\theta)/\sqrt{n}} \right| \leq c_\alpha \right) = 1 - \alpha + o(1). \label{eq:c_alpha}
\end{equation}
Following \cite{kennedy2019nonparametric}, a multiplier bootstrap \citep{van1996weak, belloni2018uniformly} can be used to obtain the above critical value $c_\alpha$.
In this approach, the distribution of the supremum in \eqref{eq:c_alpha} is approximated by the supremum of the multiplier process
\begin{equation}
    \frac{1}{\sqrt{n}} \sum_{k=1}^K \sum_{i\in \cI_k} \xi_i \cur{\frac{\phi_i(\theta; \hat \Lambda_{-k}, \hat \mu_{-k}) - \hat{\psi}_{\cf}(\theta)}{\hat \sigma_{\cf}(\theta)}}, \label{eq:multi_pro}
\end{equation}
where the multipliers $(\xi_1, \dots, \xi_n)$ are independent and identically distributed random variables with mean zero and unit variance that are independent of the observed data.
This procedure is computationally efficient because it only requires resampling the multipliers after the nuisance functions have been estimated. 
Here  we employ the Rademacher multipliers, which take values in $\{-1, 1\}$ with equal probability. Other valid choices for the multipliers include standard Gaussian random variables.
The following theorem establishes the theoretical validity of this procedure. 

\begin{thm} \label{thm:ui_const}
Let $\hat c_\alpha$ denote the $(1-\alpha)$ quantile (conditional on the sample) of the supremum of the multiplier bootstrap process, that is,
\begin{equation*}
\p\left(\left.\sup_{\theta \in \Theta} \abs{\frac{1}{\sqrt{n}} \sum_{k=1}^K \sum_{i\in \cI_k} \xi_i \cur{\frac{\phi_i(\theta; \hat \Lambda_{-k}, \hat \mu_{-k}) - \hat{\psi}_{\cf}(\theta)}{\hat \sigma_{\cf}(\theta)}}} \geq \hat c_\alpha \right| \{Y_i, U_i, L_i\}_{i=1}^n \right) = \alpha,
\end{equation*}
where $(\xi_1,\ldots,\xi_n)$ are i.i.d. Rademacher random variables independent of the sample $\{Y_i, U_i, L_i\}_{i=1}^n$.
Under the same conditions of \thmref{cf},
\begin{equation*}
\p\left(\hat \psi_{\cf}(\theta) - \frac{\hat c_\alpha \hat \sigma_{\cf}(\theta)}{\sqrt{n}}
\leq \psi(\theta) \leq \hat \psi_{\cf}(\theta) + \frac{\hat c_\alpha \hat \sigma_{\cf}(\theta)}{\sqrt{n}}, \forall \theta \in \Theta \right) = 1 - \alpha + o(1).
\end{equation*}
\end{thm}

The proof of \thmref{ui_const} is in \secref{ui_const_proof} of the Supplementary Material, by using techniques developed
in \eqref{eq:process_con} \citep{chernozhukov2014gaussian}. Given the above uniform confidence band, denoted as $ (L(\theta, \alpha), U(\theta, \alpha)) $ we can also test for example the global null hypothesis of no incremental intervention effect, 
$$
H_0: \psi(\theta) = E(Y) \text{ for all } \theta \in \Theta. 
$$
Assuming that the no-intervention regime $\theta(t,l)\equiv 1$ belongs to $\Theta$, we can test the above by checking whether the $(1-\alpha)$ confidence band contains a constant function over $\Theta$.
Equivalently, the corresponding $p$-value is given by 
$     \inf \squ{ \alpha: \inf_{\theta \in \Theta} \{ 
     U(\theta, \alpha)\} 
    < \sup_{\theta \in \Theta} \{ 
    L(\theta, \alpha) \} }  $,
or equivalently, 
     $\sup \squ{ \alpha: \inf_{\theta \in \Theta} \{ 
     U(\theta, \alpha) \} 
    \ge \sup_{\theta \in \Theta} \{ 
    L(\theta, \alpha)\} } $.

\section{Simulation}\label{sec:simu}

In this section, we investigate the finite-sample performance of the proposed estimators. We set $\tau = 2$ and generate $n$ i.i.d. copies of $\{L_i, T_i, Y_i\}$ as follows:
$
L_i \sim \Unif (0, 2),
$
$
\p(T_i > t|L_i) = \exp\{-\exp(0.2 L_i)\cdot t\},
$
$
Y_i|L_i, T_i \wedge 2 \sim \cN(3 - 0.6 L_i - (2 - T_i \wedge 2), 0.5^2). 
$
We observe $\{L_i, U_i = T_i \wedge 2, Y_i, \Delta_i = 1(U_i < \tau)\}_{i=1}^n$. 

We considered three estimators in our simulation study. The inverse probability weighted estimator proposed in \cite{ying2025incremental} is:
\begin{equation*}
\hat \psi_{\text{ipw}}(\theta) = \frac{1}{n} \sum_{i=1}^n Y_i \theta(U_i,L_i)^{\Delta_i} e^{-\int_0^{U_i}\{\theta(v,L_i)-1\}d\hat\Lambda(v|L_i)}. 
\end{equation*}
We have two versions of the AIPW estimator: $\hat \psi(\theta)$ and  $\hat \psi_{\cf}(\theta)$.
For  $\hat \psi(\theta)$ 
we employ a \cite{cox1972regression} proportional hazards model for $\hat\Lambda$, 
and a linear regression model for $\hat\mu$. 
For $\hat \psi_{\cf}(\theta)$ 
we employ splines based hazard regression (HARE) \citep{kooperberg1995hazard}  for $\hat\Lambda$ and random forests for $\hat\mu$, implemented via the \texttt{polspline} and \texttt{ranger} packages in R, respectively.
In addition, for comparison purposes we also consider 
 an oracle estimator  $\hat \psi_{\mathrm{o}}(\theta)$,  constructed using the true nuisance functions $\Lambda_0$ and $\mu_0$.

We examine the performance of the four estimators 
across eight scenarios, where $\theta(t, l) = (at + b)\exp(\beta l)$ for $(a, b, \beta)$ in
\begin{alignat*}{4}
\{& (0.9, \ 0.3, -0.7), \  && (0.9, \ 0.5, -0.7), \  && (0.7, \ 0.3, -0.5), \  && (0.7, \ 0.5, -0.5), \\
  & (0.5, \ 0.1, -0.1), \  && (0.5, \ 0.1, -0.2), \  && (0.3, \ 0.1, \phantom{-}0.4), \  && (0.3, \ 0.1, \phantom{-}0.6) \}.
\end{alignat*}
We denote these eight scenarios using $\theta_1$ - $\theta_8$ in \figref{figures/compare}. 
These choices result in a censoring rate for $T$ of approximately 20\%-30\%.
All quantities below are computed from $R=1000$ simulated data sets of sizes $n\in \{200,1000,5000\}$.

In Figure \ref{fig:figures/compare} and Tables \ref{tab:simu_n_200} - \ref{tab:simu_n_5000}  in the Supplementary Material,
we report bias, \%Bias, standard deviation  (SD), average  standard error (SE), and coverage probabilities (CP) of 95\% confidence intervals. 
For $\hat\psi_{\text{ipw}}(\theta)$ the SE is obtained using $B = 200$ Bayesian bootstrap \citep{van1996weak, kosorok2008introduction}, 
and CP of the corresponding Wald type 95\% confidence intervals. 

In \figref{figures/compare} as the sample size $n$ increases, 
bias vanishes and coverage probabilities approach nominal levels for all estimators. 
The SD panel highlights the smaller standard deviations of $\hat\psi(\theta)$ compared to $\hat\psi_{\text{ipw}}(\theta)$, illustrating the efficiency gains achieved by the efficient influence function.
\begin{figure}[ht!]
\centering
\includegraphics[width=1.0\textwidth]{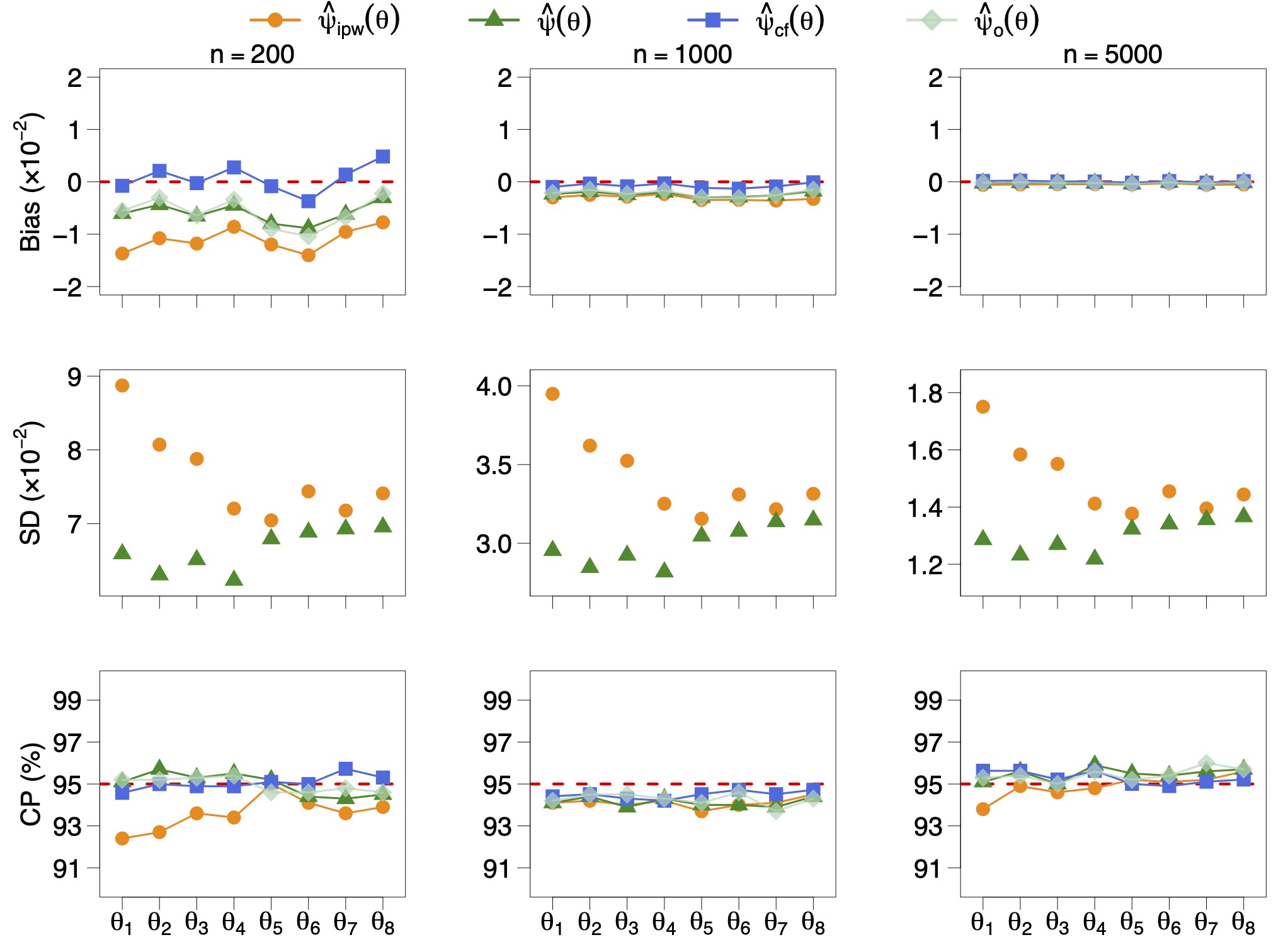}
\caption{Estimation accuracy across eight $\theta(t,l)$ specifications using simulated data. 
}
\label{fig:figures/compare}
\end{figure}

Finally we evaluate the uniform coverage performance of the proposed  confidence bands using $B=10,000$ multiplier bootstrap.
We consider $\theta(t,l) = (0.3t+0.1)\exp(\beta l)$ with $\beta \in [0.2, 0.7]$. 
In order to approximate the supremum in \eqref{eq:multi_pro} we use about 50 equally spaced grid points of $\beta$, with 0.1 distance in-between two adjacent points.
The coverage probabilities of the uniform 95\% confidence band are 94.7, 94.9 and 95.2 for sample sizes 200, 1000 and 5000, respectively. We see that they
 are satisfactory even for the smaller sample size. 

\section{Application}\label{sec:real}

We return to the motivating example of 
HPV testing among Norwegian women in the Introduction. 
For this data $T$ is the time between 2005 and 2010 from  the initial HPV negative test to the subsequent testing,
 $U$ is its censored version at the end of 2010, and the outcome $Y$ is an indicator of whether CIN2+ was detected by the end of 2010. 
\cite{roysland2025graphical} carefully examined the causal validity of such a stochastic intervention that imposes the same subsequent testing rate for all the HPV tests. They showed that  the assumption that any testing in itself does not affect the disease progression, and also  disease progression does not affect the testing regime, is likely to hold for this data. As a result  imposing the same subsequent testing rate for all the HPV tests, would make `subsequent testing' locally independent of `HPV-test type' (see their Figure 2). Hence under such a stochastic intervention any observed difference between the CN2+ detection rates can be attributed to the HPV tests with different false-negative rates. 
\cite{roysland2025graphical} performed a nonparametric IPW analysis, while in the following we apply the EIF based estimators. 

Since the original data set is not publicly available, 
our analysis uses a simulated data set from the companion website  \texttt{github.com/palryalen/paper-code} of \cite{roysland2025graphical}. 
It has features identical to those of the real data, 
and consists of 1428 individuals, with 1147 in the combined Amplicor/HC2 group and 281 in the PreTectProofer group.
For this data  the empirical average of $Y$ is 0.9\%
in the Amplicor/HC2 group and 5.7\% 
in the PreTectProofer group.  
\figref{figures/hpv_KM} shows 
that individuals in the PreTectProofer group consistently had a higher probability of subsequent testing than those in the Amplicor/HC2 group. Because this higher rate of subsequent testing is presumably due to manufacturer recommendation  \citep{roysland2025graphical}, a constant $\theta$ is considered here. 
Using a Cox proportional hazards model,
the average constant hazard ratio of subsequent testing comparing the Amplicor/HC2 group to the PreTectProofer group is estimated  to be 0.677 (95\% CI: 0.573 – 0.8). 
We note that there is no censoring in the data before the end of follow-up, therefore the average hazard ratio is well defined \citep{xu2000proportional}. 

We apply our framework by estimating the incremental causal effect of subsequent testing within the PreTectProofer group, 
and consider estimands $\psi(\theta)$ for $\theta \in [0.5, 1.1]$, representing 
an intervention that scales the hazard rate of subsequent testing in the PreTectProofer group by a factor of $\theta$. 
As in the simulation 
we consider the estimators $\hat\psi(\theta)$, $\hat\psi_{\cf}(\theta)$ and $\hat \psi_{\text{ipw}}(\theta)$. 
Since this data does not contain baseline covariates \citep{roysland2025graphical},
for $\hat\psi(\theta)$ 
we estimate $\hat\Lambda$ using log transformed Kaplan-Meier estimator and $\hat\mu$ using logistic regression for the binary $Y$. For the uniform confidence band we approximated the supreme over the interval $[0.5, 1.1]$ using about 60 
equally 0.1 spaced discrete $\theta$ values.

We display the results for $\hat\psi_{\cf}(\theta)$ in \figref{figures/hpv-psi-cf}. 
\begin{figure}[ht!]
\centering
\includegraphics[width=0.7 \textwidth]{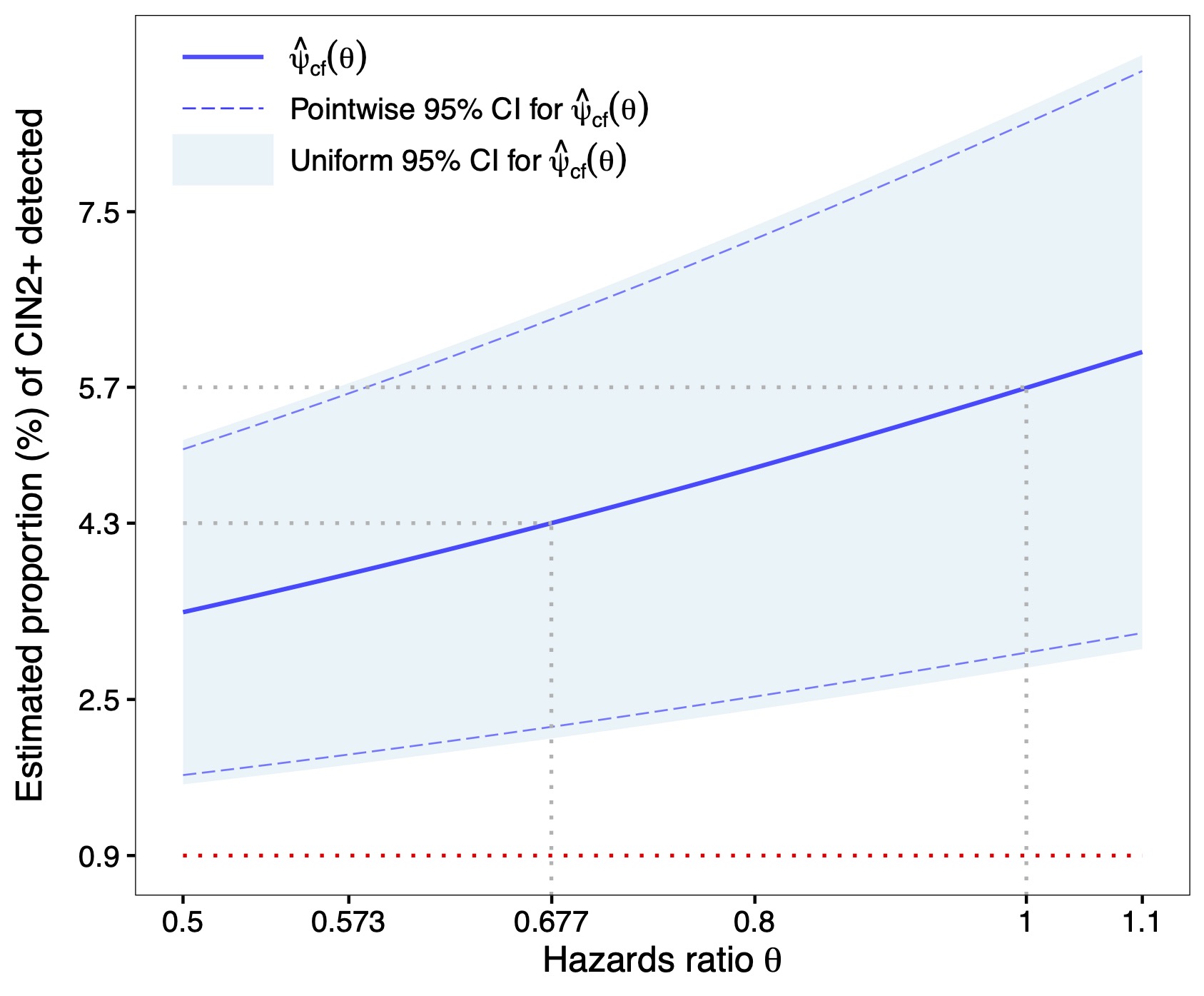}
\caption{Estimated proportion of CIN2+ detected in the PreTectProofer group under incremental intervention indexed by $\theta$.
}
\label{fig:figures/hpv-psi-cf}
\end{figure}
From \figref{figures/hpv-psi-cf}, we observe that if the hazard of the time to subsequent testing decreases, 
the proportion of CIN2+ detection also decreases from the observed 0.0569 in the PreTectProofer group; for example, 
halving the hazards would reduce the proportion 
to $\hat{\psi}_{\cf}(0.5)$= 3.4\% (95\% CI: 1.7 – 5.1\%).
Evaluating at the previously estimated hazard ratio $\theta=0.677$ gives $\hat{\psi}_{\cf}(0.677)= 4.3\%$ (95\% CI: 2.2 – 6.4\%). 
Note that the entire 95\% CI is above 0.9\%, the observed CIN2+ detection proportion in the combined Amplicor/HC2 group.
In fact, the whole confidence band in the figure is above 0.9\%. 
This suggests that even under the same subsequent testing regime, 
the CIN2+ detection rate is still significantly higher in the PreTectProofer group. 
Therefore it is likely that  the PreTectProofer test has  higher false-negative rate compared to the other two tests.  

As a side note the uniform confidence band is not much wider than the pointwise CI's mainly because the EIF in \eqref{eq:eifconst} is very smooth in $\theta$.  
\figref{figures/hpv-psi} of the Supplementary Material contains the results from  all three estimators, where the point estimates are close to each other, with $\hat\psi(\theta)$ and $\hat\psi_{\cf}(\theta)$ exhibiting slightly narrower confidence intervals than $\hat \psi_{\text{ipw}}(\theta)$. 
Also as a comparison, 
applying the GitHub \texttt{R} code from \cite{roysland2025graphical} to the same simulated data set yields an estimated detection rate of  4.2\% (95\% CI: 2.2 – 6.2\%) for the PreTectProofer test under the stochastic intervention, which is very close to our estimate of 4.3\% (95\% CI: 2.2 – 6.4\%). 
Note that \cite{roysland2025graphical} used a smoothed $\theta(t)$  over time and, in doing so, were not able to obtain any uniform confidence bands.




\

\section*{Acknowledgement}

This work was presented at the 2025 Joint Statistical Meetings and the 2026 American Causal Inference Conference.

\bibliographystyle{agsm}
\bibliography{ref-sorted}

\newpage
\suppformatappendix

\section*{Supplementary Material}

\addcontentsline{toc}{part}{Supplementary Material} 
\etocsettocstyle{\section*{Contents}}{}
\etocsetnexttocdepth{subsection}
\localtableofcontents 

\newpage

\section{Additional Material for Simulation and Application} \label{sec:add_simu}

\subsection{Additional Material for Simulation}

This section includes the tables that display the simulation results in Section \ref{sec:simu} of the main paper.

\begin{table}[ht!]
\caption{\label{tab:simu_n_200}Simulation results for different estimators. Each observed dataset has sample size 200, and 1000 datasets are simulated. Bias, SD, and SE are scaled by $10^{-2}$.}
\centering
\begin{tabular}{lllrrccc}
\toprule
$\theta(t,l)$ & $\psi(\theta)$ & Estimator & Bias & \%Bias & SD & SE & CP \\
\midrule
\multirow{4}{*}{$\theta_{1}$} & \multirow{4}{*}{1.655} & $\hat\psi_{\text{ipw}}(\theta)$ & -1.370 & -0.828 & 8.872 & 8.319 & 92.4 \\
 &  & $\hat\psi(\theta)$ & -0.605 & -0.365 & 6.592 & 6.410 & 95.1 \\
 &  & $\hat\psi_{\text{o}}(\theta)$ & -0.552 & -0.334 & 6.604 & 6.509 & 95.2 \\
 &  & $\hat\psi_{\text{cf}}(\theta)$ & -0.076 & -0.046 & 8.264 & 6.841 & 94.6 \\
\midrule
\multirow{4}{*}{$\theta_{2}$} & \multirow{4}{*}{1.546} & $\hat\psi_{\text{ipw}}(\theta)$ & -1.077 & -0.697 & 8.071 & 7.617 & 92.7 \\
 &  & $\hat\psi(\theta)$ & -0.436 & -0.282 & 6.306 & 6.190 & 95.7 \\
 &  & $\hat\psi_{\text{o}}(\theta)$ & -0.304 & -0.197 & 6.313 & 6.265 & 95.2 \\
 &  & $\hat\psi_{\text{cf}}(\theta)$ & 0.209 & 0.135 & 7.144 & 6.405 & 95.0 \\
\midrule
\multirow{4}{*}{$\theta_{3}$} & \multirow{4}{*}{1.635} & $\hat\psi_{\text{ipw}}(\theta)$ & -1.179 & -0.721 & 7.878 & 7.497 & 93.6 \\
 &  & $\hat\psi(\theta)$ & -0.655 & -0.401 & 6.512 & 6.371 & 95.3 \\
 &  & $\hat\psi_{\text{o}}(\theta)$ & -0.645 & -0.395 & 6.504 & 6.453 & 95.3 \\
 &  & $\hat\psi_{\text{cf}}(\theta)$ & -0.025 & -0.015 & 7.323 & 6.657 & 94.9 \\
\midrule
\multirow{4}{*}{$\theta_{4}$} & \multirow{4}{*}{1.505} & $\hat\psi_{\text{ipw}}(\theta)$ & -0.859 & -0.571 & 7.205 & 6.930 & 93.4 \\
 &  & $\hat\psi(\theta)$ & -0.447 & -0.297 & 6.233 & 6.164 & 95.5 \\
 &  & $\hat\psi_{\text{o}}(\theta)$ & -0.339 & -0.225 & 6.224 & 6.215 & 95.4 \\
 &  & $\hat\psi_{\text{cf}}(\theta)$ & 0.275 & 0.183 & 6.506 & 6.244 & 94.9 \\
\midrule
\multirow{4}{*}{$\theta_{5}$} & \multirow{4}{*}{1.729} & $\hat\psi_{\text{ipw}}(\theta)$ & -1.197 & -0.692 & 7.045 & 6.881 & 95.0 \\
 &  & $\hat\psi(\theta)$ & -0.799 & -0.462 & 6.795 & 6.681 & 95.2 \\
 &  & $\hat\psi_{\text{o}}(\theta)$ & -0.897 & -0.519 & 6.757 & 6.755 & 94.6 \\
 &  & $\hat\psi_{\text{cf}}(\theta)$ & -0.083 & -0.048 & 7.618 & 7.076 & 95.1 \\
\midrule
\multirow{4}{*}{$\theta_{6}$} & \multirow{4}{*}{1.774} & $\hat\psi_{\text{ipw}}(\theta)$ & -1.402 & -0.790 & 7.438 & 7.168 & 94.1 \\
 &  & $\hat\psi(\theta)$ & -0.885 & -0.499 & 6.886 & 6.722 & 94.4 \\
 &  & $\hat\psi_{\text{o}}(\theta)$ & -1.042 & -0.588 & 6.858 & 6.821 & 94.6 \\
 &  & $\hat\psi_{\text{cf}}(\theta)$ & -0.370 & -0.209 & 8.538 & 7.273 & 95.0 \\
\midrule
\multirow{4}{*}{$\theta_{7}$} & \multirow{4}{*}{1.652} & $\hat\psi_{\text{ipw}}(\theta)$ & -0.955 & -0.578 & 7.178 & 7.162 & 93.6 \\
 &  & $\hat\psi(\theta)$ & -0.620 & -0.375 & 6.929 & 6.936 & 94.3 \\
 &  & $\hat\psi_{\text{o}}(\theta)$ & -0.678 & -0.410 & 6.898 & 6.979 & 94.8 \\
 &  & $\hat\psi_{\text{cf}}(\theta)$ & 0.138 & 0.084 & 7.363 & 7.365 & 95.7 \\
\midrule
\multirow{4}{*}{$\theta_{8}$} & \multirow{4}{*}{1.553} & $\hat\psi_{\text{ipw}}(\theta)$ & -0.774 & -0.499 & 7.410 & 7.389 & 93.9 \\
 &  & $\hat\psi(\theta)$ & -0.303 & -0.195 & 6.956 & 7.025 & 94.5 \\
 &  & $\hat\psi_{\text{o}}(\theta)$ & -0.222 & -0.143 & 6.938 & 7.048 & 94.6 \\
 &  & $\hat\psi_{\text{cf}}(\theta)$ & 0.483 & 0.311 & 7.064 & 7.400 & 95.3 \\
\bottomrule
\end{tabular}
\end{table}

\newpage
\begin{table}[ht!]
\caption{\label{tab:simu_n_1000}Simulation results for different estimators. Each observed dataset has sample size 1000, and 1000 datasets are simulated. Bias, SD, and SE are scaled by $10^{-2}$.}
\centering
\begin{tabular}{lllrrccc}
\toprule
$\theta(t,l)$ & $\psi(\theta)$ & Estimator & Bias & \%Bias & SD & SE & CP \\
\midrule
\multirow{4}{*}{$\theta_{1}$} & \multirow{4}{*}{1.655} & $\hat\psi_{\text{ipw}}(\theta)$ & -0.298 & -0.180 & 3.949 & 3.840 & 94.1 \\
 &  & $\hat\psi(\theta)$ & -0.235 & -0.142 & 2.954 & 2.910 & 94.1 \\
 &  & $\hat\psi_{\text{o}}(\theta)$ & -0.213 & -0.129 & 2.935 & 2.910 & 94.2 \\
 &  & $\hat\psi_{\text{cf}}(\theta)$ & -0.098 & -0.059 & 3.352 & 3.008 & 94.4 \\
\midrule
\multirow{4}{*}{$\theta_{2}$} & \multirow{4}{*}{1.546} & $\hat\psi_{\text{ipw}}(\theta)$ & -0.249 & -0.161 & 3.620 & 3.496 & 94.2 \\
 &  & $\hat\psi(\theta)$ & -0.189 & -0.122 & 2.846 & 2.797 & 94.4 \\
 &  & $\hat\psi_{\text{o}}(\theta)$ & -0.152 & -0.098 & 2.832 & 2.797 & 94.5 \\
 &  & $\hat\psi_{\text{cf}}(\theta)$ & -0.034 & -0.022 & 3.006 & 2.848 & 94.5 \\
\midrule
\multirow{4}{*}{$\theta_{3}$} & \multirow{4}{*}{1.635} & $\hat\psi_{\text{ipw}}(\theta)$ & -0.284 & -0.173 & 3.524 & 3.428 & 94.0 \\
 &  & $\hat\psi(\theta)$ & -0.247 & -0.151 & 2.924 & 2.883 & 93.9 \\
 &  & $\hat\psi_{\text{o}}(\theta)$ & -0.232 & -0.142 & 2.908 & 2.884 & 94.5 \\
 &  & $\hat\psi_{\text{cf}}(\theta)$ & -0.086 & -0.053 & 3.043 & 2.936 & 94.3 \\
\midrule
\multirow{4}{*}{$\theta_{4}$} & \multirow{4}{*}{1.505} & $\hat\psi_{\text{ipw}}(\theta)$ & -0.231 & -0.154 & 3.251 & 3.147 & 94.2 \\
 &  & $\hat\psi(\theta)$ & -0.198 & -0.132 & 2.818 & 2.776 & 94.3 \\
 &  & $\hat\psi_{\text{o}}(\theta)$ & -0.165 & -0.109 & 2.808 & 2.776 & 94.3 \\
 &  & $\hat\psi_{\text{cf}}(\theta)$ & -0.030 & -0.020 & 2.856 & 2.793 & 94.2 \\
\midrule
\multirow{4}{*}{$\theta_{5}$} & \multirow{4}{*}{1.729} & $\hat\psi_{\text{ipw}}(\theta)$ & -0.348 & -0.201 & 3.156 & 3.101 & 93.7 \\
 &  & $\hat\psi(\theta)$ & -0.300 & -0.173 & 3.046 & 3.020 & 94.0 \\
 &  & $\hat\psi_{\text{o}}(\theta)$ & -0.302 & -0.174 & 3.031 & 3.021 & 94.1 \\
 &  & $\hat\psi_{\text{cf}}(\theta)$ & -0.116 & -0.067 & 3.084 & 3.062 & 94.5 \\
\midrule
\multirow{4}{*}{$\theta_{6}$} & \multirow{4}{*}{1.774} & $\hat\psi_{\text{ipw}}(\theta)$ & -0.345 & -0.194 & 3.310 & 3.249 & 94.0 \\
 &  & $\hat\psi(\theta)$ & -0.290 & -0.163 & 3.076 & 3.049 & 94.0 \\
 &  & $\hat\psi_{\text{o}}(\theta)$ & -0.307 & -0.173 & 3.058 & 3.051 & 94.6 \\
 &  & $\hat\psi_{\text{cf}}(\theta)$ & -0.130 & -0.073 & 3.173 & 3.118 & 94.7 \\
\midrule
\multirow{4}{*}{$\theta_{7}$} & \multirow{4}{*}{1.652} & $\hat\psi_{\text{ipw}}(\theta)$ & -0.359 & -0.217 & 3.216 & 3.216 & 94.1 \\
 &  & $\hat\psi(\theta)$ & -0.257 & -0.156 & 3.137 & 3.126 & 93.9 \\
 &  & $\hat\psi_{\text{o}}(\theta)$ & -0.254 & -0.154 & 3.128 & 3.127 & 93.7 \\
 &  & $\hat\psi_{\text{cf}}(\theta)$ & -0.089 & -0.054 & 3.158 & 3.164 & 94.5 \\
\midrule
\multirow{4}{*}{$\theta_{8}$} & \multirow{4}{*}{1.553} & $\hat\psi_{\text{ipw}}(\theta)$ & -0.326 & -0.210 & 3.314 & 3.326 & 94.5 \\
 &  & $\hat\psi(\theta)$ & -0.178 & -0.115 & 3.148 & 3.156 & 94.4 \\
 &  & $\hat\psi_{\text{o}}(\theta)$ & -0.151 & -0.097 & 3.142 & 3.157 & 94.3 \\
 &  & $\hat\psi_{\text{cf}}(\theta)$ & -0.009 & -0.006 & 3.161 & 3.193 & 94.7 \\
\bottomrule
\end{tabular}
\end{table}

\newpage
\begin{table}[ht!]
\caption{\label{tab:simu_n_5000}Simulation results for different estimators. Each observed dataset has sample size 5000, and 1000 datasets are simulated. Bias, SD, and SE are scaled by $10^{-2}$.}
\centering
\begin{tabular}{lllrrccc}
\toprule
$\theta(t,l)$ & $\psi(\theta)$ & Estimator & Bias & \%Bias & SD & SE & CP \\
\midrule
\multirow{4}{*}{$\theta_{1}$} & \multirow{4}{*}{1.655} & $\hat\psi_{\text{ipw}}(\theta)$ & -0.059 & -0.036 & 1.751 & 1.723 & 93.8 \\
 &  & $\hat\psi(\theta)$ & -0.025 & -0.015 & 1.286 & 1.305 & 95.1 \\
 &  & $\hat\psi_{\text{o}}(\theta)$ & -0.019 & -0.012 & 1.288 & 1.305 & 95.3 \\
 &  & $\hat\psi_{\text{cf}}(\theta)$ & 0.016 & 0.010 & 1.328 & 1.324 & 95.6 \\
\midrule
\multirow{4}{*}{$\theta_{2}$} & \multirow{4}{*}{1.546} & $\hat\psi_{\text{ipw}}(\theta)$ & -0.051 & -0.033 & 1.584 & 1.570 & 94.9 \\
 &  & $\hat\psi(\theta)$ & -0.018 & -0.011 & 1.232 & 1.254 & 95.6 \\
 &  & $\hat\psi_{\text{o}}(\theta)$ & -0.010 & -0.006 & 1.233 & 1.253 & 95.4 \\
 &  & $\hat\psi_{\text{cf}}(\theta)$ & 0.022 & 0.014 & 1.254 & 1.265 & 95.6 \\
\midrule
\multirow{4}{*}{$\theta_{3}$} & \multirow{4}{*}{1.635} & $\hat\psi_{\text{ipw}}(\theta)$ & -0.047 & -0.029 & 1.551 & 1.539 & 94.6 \\
 &  & $\hat\psi(\theta)$ & -0.027 & -0.017 & 1.269 & 1.293 & 95.0 \\
 &  & $\hat\psi_{\text{o}}(\theta)$ & -0.023 & -0.014 & 1.271 & 1.293 & 95.0 \\
 &  & $\hat\psi_{\text{cf}}(\theta)$ & 0.007 & 0.004 & 1.292 & 1.304 & 95.2 \\
\midrule
\multirow{4}{*}{$\theta_{4}$} & \multirow{4}{*}{1.505} & $\hat\psi_{\text{ipw}}(\theta)$ & -0.046 & -0.031 & 1.412 & 1.413 & 94.8 \\
 &  & $\hat\psi(\theta)$ & -0.028 & -0.018 & 1.218 & 1.244 & 95.9 \\
 &  & $\hat\psi_{\text{o}}(\theta)$ & -0.021 & -0.014 & 1.219 & 1.243 & 95.6 \\
 &  & $\hat\psi_{\text{cf}}(\theta)$ & 0.006 & 0.004 & 1.230 & 1.249 & 95.6 \\
\midrule
\multirow{4}{*}{$\theta_{5}$} & \multirow{4}{*}{1.729} & $\hat\psi_{\text{ipw}}(\theta)$ & -0.054 & -0.031 & 1.377 & 1.393 & 95.2 \\
 &  & $\hat\psi(\theta)$ & -0.041 & -0.024 & 1.323 & 1.354 & 95.5 \\
 &  & $\hat\psi_{\text{o}}(\theta)$ & -0.042 & -0.024 & 1.324 & 1.354 & 95.2 \\
 &  & $\hat\psi_{\text{cf}}(\theta)$ & -0.016 & -0.009 & 1.346 & 1.362 & 95.0 \\
\midrule
\multirow{4}{*}{$\theta_{6}$} & \multirow{4}{*}{1.774} & $\hat\psi_{\text{ipw}}(\theta)$ & -0.026 & -0.015 & 1.455 & 1.459 & 95.1 \\
 &  & $\hat\psi(\theta)$ & -0.010 & -0.006 & 1.341 & 1.368 & 95.4 \\
 &  & $\hat\psi_{\text{o}}(\theta)$ & -0.013 & -0.007 & 1.342 & 1.368 & 95.4 \\
 &  & $\hat\psi_{\text{cf}}(\theta)$ & 0.014 & 0.008 & 1.369 & 1.381 & 94.9 \\
\midrule
\multirow{4}{*}{$\theta_{7}$} & \multirow{4}{*}{1.652} & $\hat\psi_{\text{ipw}}(\theta)$ & -0.059 & -0.036 & 1.395 & 1.449 & 95.2 \\
 &  & $\hat\psi(\theta)$ & -0.038 & -0.023 & 1.356 & 1.401 & 95.6 \\
 &  & $\hat\psi_{\text{o}}(\theta)$ & -0.038 & -0.023 & 1.356 & 1.401 & 96.0 \\
 &  & $\hat\psi_{\text{cf}}(\theta)$ & -0.014 & -0.009 & 1.392 & 1.411 & 95.1 \\
\midrule
\multirow{4}{*}{$\theta_{8}$} & \multirow{4}{*}{1.553} & $\hat\psi_{\text{ipw}}(\theta)$ & -0.052 & -0.033 & 1.444 & 1.500 & 95.6 \\
 &  & $\hat\psi(\theta)$ & -0.016 & -0.010 & 1.366 & 1.414 & 95.7 \\
 &  & $\hat\psi_{\text{o}}(\theta)$ & -0.011 & -0.007 & 1.366 & 1.414 & 95.7 \\
 &  & $\hat\psi_{\text{cf}}(\theta)$ & 0.012 & 0.008 & 1.401 & 1.423 & 95.2 \\
\bottomrule
\end{tabular}
\end{table}

\clearpage

\subsection{Additional Material for Application} \label{sec:add_real}

\figref{figures/hpv-psi} presents the application results from \secref{real} of the main paper.
\begin{figure}[ht!]
\centering
\includegraphics[width= \textwidth]{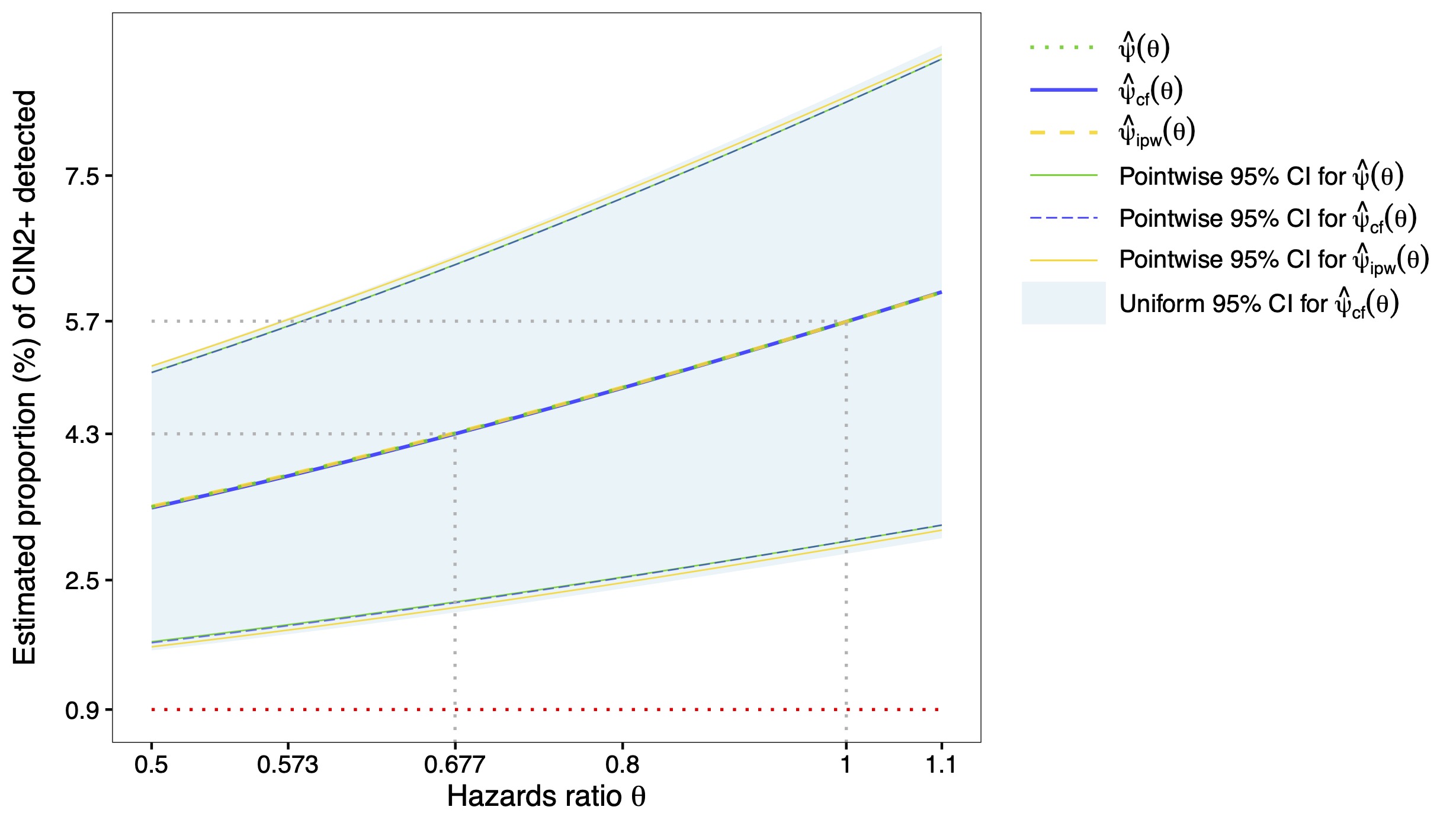}
\caption{Estimated proportion of CIN2+ detected in the PreTectProofer group under incremental intervention indexed by $\theta$.
}
\label{fig:figures/hpv-psi}
\end{figure}

\section{Proofs of Theorems} \label{sec:thm_proofs}

\subsection{Notation} \label{sec:notation}

We introduce the notation used throughout the proofs.
\begin{center}
\begin{tabular}{l @{\hspace{1cm}} l}
\toprule
Notation & Description \\
\midrule
$C$ & Generic positive constants that may vary from place to place \\
$a_n \lesssim b_n$ & For two positive sequences $a_n$ and $b_n$, $a_n \leq C b_n$ for all $n$ \\
$P$ & Expectation with respect to the observed data distribution \\
$P_n$ & Empirical average, e.g., $P_n \{\phi(\theta; \cdot, \cdot)\} = \sum_{i=1}^n \phi_i(\theta; \cdot, \cdot) / n$ \\
$\sup_{t, l} |f(t,l)|$ & For a real-valued function $f$, $\sup_{t, l} |f(t,l)| = \sup_{t \in [0, \tau], l \in \cL} |f(t,l)|$ \\
\bottomrule
\end{tabular}
\end{center}


\noindent Let $(\cF, \|\cdot\|)$ be a class of real-valued functions $f:\cX \to \bbR$, where $\|\cdot\|$ denotes the $L_2(P)$-norm with respect to a probability measure $P$ on $\cX$.

\begin{center}
\begin{tabular}{l @{\hspace{1cm}} l}
\toprule
Notation & Description \\
\midrule
$F$ & Envelope function of $\cF$ such that $|f(x)| \leq F(x)$ for all $x \in \cX$ and $f \in \cF$ \\
Bracket $[l,u]$ & Set of all functions $f$ with $l \le f \le u$ \\
$\epsilon$-$L_2(P)$-bracket & Bracket $[l, u]$ with $\|u - l\| < \epsilon$ \\
$N_{[]}(\epsilon, \cF, L_2(P))$ & $L_2(P)$-bracketing number of $\cF$, \\
& the minimum number of $\epsilon$-$L_2(P)$-brackets needed to cover $\cF$ \\
$J_{[]}(\cF)$ & $L_2(P)$-bracketing integral of $\cF$, $J_{[]}(\cF) = \int_0^\infty \sqrt{\log N_{[]}(\epsilon, \cF, L_2(P))} d\epsilon$ \\
\bottomrule
\end{tabular}
\end{center}

\subsection{Proof of \thmref{eifiden}} \label{sec:pfeifiden}

\begin{proof}
Denote $O \coloneqq (Y, U, L)$ and $\P$ as the distribution of $O$. Recall $\Delta = 1 (U < \tau)$. Note that $\theta(U, L)^\Delta = \theta(T, L)^\Delta$, and $\Delta = 1 (U < \tau) = 1 (T < \tau)$. For a fixed function $\theta$, define 
\begin{equation}
g(O, \P) \coloneqq Y \theta(U,L)^{\Delta}e^{-\int_0^{U-}\{\theta(v,L) - 1\} d\Lambda_0(v|L)} = Y \theta(T,L)^{\Delta}e^{-\int_0^{T \wedge \tau}\{\theta(t,L) - 1\}d\Lambda_0(t|L)}, \label{eq:u-}
\end{equation}
where $\Lambda_0(t|L)$ denotes the cumulative hazard function of $T$ given $L$, as defined in Section~\ref{sec:iden}. 
In \eqref{eq:u-}, using the upper limit $U-$ ensures that individuals treated exactly at $\tau$ are classified as untreated, and is essential to correctly derive the efficient influence function in the presence of a point mass at $\tau$.
Since the cumulative hazard functions of $U$ and $T$ given $L$ are identical on $[0, \tau)$, we interpret $\Lambda_0(v|L)$ in $g(O, \P)$ as the cumulative hazard function of $U$ at time $v$ given $L$.

Our estimand is 
$
\psi(\theta) 
= E_{\P} \left\{g(O, \P) \right\}.
$
Following \cite{hines2022demystifying}, to obtain the estimand's efficient influence function, we will perturb the estimand in the direction $\widetilde{\P}$ of a point mass on a single observation $\tilde o = (\tilde{y}, \tilde u, \tilde{l})$. This pertubation is parameterized via the one-dimensional mixture model
$$
\P_\gamma=\gamma \widetilde{\P}+(1-\gamma) \P,
$$
indexed by $\gamma \in [0,1]$.

The estimand's efficient influence function is $\phi(O, \P)$ such that
$$\left.\frac{d}{d \gamma} E_{\P_\gamma} \left\{g\left(O, \P_\gamma\right)\right\} \right|_{\gamma=0}
=E_{\widetilde{\P}} \left\{\phi(O, \P) \right\},$$
and \cite{hines2022demystifying} shows the following equality which simplifies the derivation of $\phi(O, \P)$, 
\begin{equation*}
\left.\frac{d}{d \gamma} E_{\P_\gamma} \left\{g\left(O, \P_\gamma\right)\right\} \right|_{\gamma=0}
= \left.E_{\P}\left\{\frac{d}{d \gamma} g(o, \P_\gamma)\right\}\right|_{\gamma=0}+g(\tilde o, \P)-E_{\P}\{g(O, \P)\}.
\end{equation*}
The goal is then to write the right-hand side (RHS) of the above as $E_{\widetilde{\P}} (\cdot)$; note that the last two terms can immediately be expressed as
\begin{equation*}
g(\tilde{o}, \P) = \tilde y \theta(\tilde u, \tilde l)^{\tilde \delta} e^{-\int_0^{\tilde u-}\{\theta(v, \tilde l)-1\} d \Lambda_0(v | \tilde l)} 
= E_{\widetilde{\P}} \left[ Y \theta(U,L)^{\Delta}e^{-\int_0^U\{\theta(v,L) - 1\} d\Lambda_0(v|L)} \right],
\end{equation*}
and
$$
E_{\P}\{g(O, \P)\} = \psi(\theta) = E_{\widetilde{\P}} \{ \psi(\theta) \}.
$$
We are left to express $\left.E_{\P}\left\{\frac{d}{d \gamma} g(o, \P_\gamma)\right\}\right|_{\gamma=0}$ in a similar way.

Below, we denote by $f(y, u, l)$ the joint density of $Y = y$, $U = u$ and $L = l$, $f(y | u, l)$ the conditional density of $Y = y$ given $U = u$ and $L = l$, by $f(u, l)$ the joint density of $U = u$ and $L = l$, by $f(u | l)$ the conditional density of $U = u$ given $L = l$, and by $f(l)$ the marginal density of $L = l$, all with respect to the distribution $\P$. Analogously, we use $f_\gamma(\cdot)$ to denote the corresponding densities under the distribution $\P_\gamma$.

Let ${\1}_{\tilde {o}}(o)$ denote the Dirac delta function with respect to $\tilde {o}$, defined as the density of a point mass at $\tilde {o}$: it is zero everywhere except at $\tilde {o}$, and its integral over the real line is equal to one.

The mixture model specifies $\f(y, u, l) = \gamma \1_{(\tilde y, \tilde u, \tilde l)}(y, u, l) + (1-\gamma) f(y, u, l)$, and integrating it with respect to $y$ and then $u$ yields $\f(u, l) = \gamma \1_{(\tilde u, \tilde l)}(u, l) + (1-\gamma) f(u,l)$ and $\f(l) = \gamma \1_{\tilde l}(l) + (1-\gamma) f(l)$. 
Immediate useful facts are $\1_{(\tilde u, \tilde l)}(u, l) = \1_{\tilde u}(u)\1_{\tilde l}(l) $, and 
$$
\frac{d}{d \gamma}f_\gamma(u, l) = \1_{(\tilde u, \tilde l)}(u, l) - f(u,l), \quad \frac{d}{d \gamma}f_\gamma(l) = \1_{\tilde l}(l) - f(l).
$$

Note that
$$
f_\gamma(u | l) = \frac{f_\gamma(u,l)}{\f(l)} = \frac{\gamma \1_{(\tilde u, \tilde l)}(u, l) + (1-\gamma) f(u,l)}{\f(l)},
$$
then
\begin{align*}
\left.\frac{d}{d \gamma}f_\gamma(u | l)\right|_{\gamma=0}
&= \left. \left[\frac{\1_{(\tilde u, \tilde l)}(u, l) - f(u,l)}{\f(l)} - \frac{\{\gamma \1_{(\tilde u, \tilde l)}(u, l) + (1-\gamma) f(u,l)\} \frac{d}{d \gamma}f_\gamma(l)}{\f(l)^2}\right] \right|_{\gamma=0} \\
&= \frac{\1_{(\tilde u, \tilde l)}(u, l) - f(u,l)}{f(l)} - \frac{f(u,l) \{\1_{\tilde l}(l) - f(l)\}}{f(l)^2} = \frac{\1_{\tilde l}(l)}{f(l)} \{\1_{\tilde u}(u) - f(u | l)\}.
\end{align*}
Denote $S_\gamma(u | l) = \P_\gamma(U \geq u |L=l) = 1-\int_0^{u-} f_\gamma(s | l) d s$. We have 
\begin{align*}
\left.\d S_\gamma(u | l)\right|_{\gamma=0} &= - \left. \int_0^{u-} \frac{d}{d \gamma}f_\gamma(s | l) d s\right|_{\gamma=0} 
= -\frac{\1_{\tilde l}(l)}{f(l)} \int_0^{u-} \{\1_{\tilde u}(s) - f(s | l)\} d s \\
&= -\frac{\1_{\tilde l}(l)}{f(l)} [1(\tilde u < u) - \{1-S(u|l)\}] = \frac{\1_{\tilde l}(l)}{f(l)} \{1(u \leq \tilde u) - S(u|l)\}.
\end{align*}
Using the two derivatives above, we obtain
\begin{align*}
\left.\frac{d}{d \gamma} \lambda_\gamma(u | l)\right|_{\gamma=0} 
&= \left.\frac{d}{d \gamma} \frac{f_\gamma(u | l)}{S_\gamma(u | l)}\right|_{\gamma=0} 
= \frac{\1_{\tilde l}(l)}{f(l)} \left[ \frac{\1_{\tilde u}(u) - f(u | l)}{S(u|l)} - \frac{f(u | l)}{S(u|l)^2} \{1(u \leq \tilde u) - S(u|l)\} \right] \\
&= \frac{\1_{\tilde l}(l)}{f(l)} \left\{ \frac{\1_{\tilde u}(u)}{S(u|l)} - \frac{f(u | l)}{S(u|l)^2} 1(u \leq \tilde u) \right\}.
\end{align*}
As defined in Section~\ref{sec:iden}, for $u\in[0,\tau)$, the hazard function of $U$ at time $u$ given $L = l$ under the distribution $\P_\gamma$ is $\lambda_\gamma(u|l) = {\f(u|l)}/{S_\gamma(u | l)}$. We have 
\begin{align}
\left.\frac{d}{d \gamma} g(o, \P_\gamma) \right|_{\gamma=0} 
&= \left.\frac{d}{d \gamma} \left[ y \theta(u, l)^\delta e^{-\int_0^{u-} \{\theta(v, l)-1\} d \Lambda_\gamma(v | l)} \right] \right|_{\gamma=0} \notag \\
&= y \theta(u, l)^\delta e^{-\int_0^u \{\theta(v, l)-1\} d \Lambda_0(v | l)} \cdot \left[ -\left.\int_0^{u-} \{\theta(v, l)-1\} \frac{d}{d \gamma} \lambda_\gamma(v | l) \right|_{\gamma=0} d v \right] \notag \\
&= y \theta(u, l)^\delta e^{-\int_0^u\{\theta(v, l)-1\} d \Lambda_0(v | l)} \cdot \int_0^{u-} \{\theta(v, l)-1\} (T_1 - T_2) d v \cdot \frac{\1_{\tilde l}(l)}{f(l)} , \label{eq:derivative}
\end{align}
where
$$
T_1 \coloneqq \frac{f(v | l)}{S(v|l)^2} 1(v \leq \tilde u) , \quad
T_2 \coloneqq \frac{\1_{\tilde u}(v)}{S(v|l)}.
$$
For the second factor in \eqref{eq:derivative}, we compute 
\begin{equation*}
\int_0^{u-}\{\theta(v, l)-1\} T_1 d v 
= \int_0^{u-} \frac{\{\theta(v, l)-1\} f(v | l)}{S(v|l)^2} 1(v \leq \tilde u) d v 
= \int_0^{\tilde u \wedge u} \frac{\{\theta(v, l)-1\} f(v | l)}{S(v|l)^2} d v ,
\end{equation*}
and
\begin{equation*}
\int_0^{u-}\{\theta(v, l)-1\} T_2 d v 
= \int_0^{u-} \frac{\theta(v, l)-1}{S(v|l)} \1_{\tilde u}(v) d v 
= \frac{\theta(\tilde{u}, l)-1}{S(\tilde{u} | l)} 1(\tilde u < u).
\end{equation*} 
The difference between the above two terms is 
\begin{align*}
\int_0^{u-}\{\theta(v, l)-1\} (T_1 - T_2) d v 
= \int_0^{\tilde u \wedge u} \frac{\{\theta(v, l)-1\} f(v | l)}{S(v|l)^2} d v 
- \frac{\{\theta(\tilde{u}, l)-1\} 1(\tilde u < u)}{S(\tilde{u} | l)} .
\end{align*}
Then 
\begin{align*}
& \left.E_{\P}\left\{\frac{d}{d \gamma} g(o,\P_\gamma)\right\}\right|_{\gamma=0} 
= E_{\P}\left\{\left.\frac{d}{d \gamma} g(o,\P_\gamma)\right|_{\gamma=0}\right\} \\
&= \int y \theta(u, l)^\delta e^{-\int_0^u\{\theta(v, l)-1\} d \Lambda_0(v | l)} 
\left[ \int_0^{\tilde u \wedge u} \frac{\{\theta(v, l)-1\} f(v | l)}{S(v|l)^2} dv \right] \frac{\1_{\tilde l}(l)}{f(l)} f(y | u, l) f(u|l) f(l) dl dy du \\
&\quad - \int y \theta(u, l)^\delta e^{-\int_0^u\{\theta(v, l)-1\} d \Lambda_0(v | l)} \frac{\{\theta(\tilde{u}, l) - 1\} 1(\tilde u < u)}{S(\tilde{u} | l)} \frac{\1_{\tilde l}(l)}{f(l)} f(y | u, l) f(u|l) f(l) dl dy du .
\end{align*}
The first term gives
\begin{align*}
& \int y \theta(u, l)^\delta e^{-\int_0^u\{\theta(v, l)-1\} d \Lambda_0(v | l)} 
\left[\int_0^{\tilde u \wedge u} \frac{\{\theta(v, l)-1\} f(v | l)}{S(v|l)^2} dv\right] \frac{\1_{\tilde l}(l)}{f(l)} f(y | u, l) f(u|l) f(l) dl dy du \\
&= \int y \theta(u, \tilde l)^\delta e^{-\int_0^u\{\theta(v, \tilde l)-1\} d \Lambda_0(v | \tilde l)} 
\left[ \int_0^{\tilde u \wedge u} \frac{\{\theta(v, \tilde l)-1\}}{S(v|\tilde l)} d \Lambda_0(v | \tilde l) \right] 
f(y | u, \tilde l) f(u |\tilde l) dy du \\
&= \int E(Y | u, \tilde l) \theta(u, \tilde l)^\delta e^{-\int_0^u\{\theta(v, \tilde l)-1\} d \Lambda_0(v | \tilde l)} \int_0^{\tilde u \wedge u} \frac{\{\theta(v, \tilde l)-1\}}{S(v|\tilde l)} d \Lambda_0(v | \tilde l) f(u |\tilde l) du \\
&= E \left[ \left. E(Y | U', \tilde l) \theta(U', \tilde l)^{\Delta'} e^{-\int_0^{U'}\{\theta(v, \tilde l)-1\} d \Lambda_0(v | \tilde l)} \int_0^{\tilde u \wedge U'} \frac{\{\theta(v, \tilde l)-1\}}{S(v|\tilde l)} d\Lambda_0(v | \tilde l) \right| \tilde l \right] \\
&= E_{\widetilde{\P}}\left[\left. E \left[ \left. E(Y | U', L) \theta(U', L)^{\Delta'} e^{-\int_0^{U'}\{\theta(v, L)-1\} d \Lambda_0(v | L)} \int_0^{u \wedge U'} \frac{\{\theta(v, L)-1\}}{S(v|L)} d\Lambda_0(v | L) \right| L \right] \right|_{u=U} \right].
\end{align*}
The second term gives 
\begin{align*}
& \int
y \theta(u, l)^\delta e^{-\int_0^u\{\theta(v, l)-1\} d \Lambda_0(v | l)} \frac{\{\theta(\tilde{u}, l) - 1\} 1(\tilde u < u)}{S(\tilde{u} | l)} \frac{\1_{\tilde l}(l)}{f(l)} f(y | u, l) f(u|l) f(l) dl dy du \\
&= \int y \theta(u, \tilde l)^\delta e^{-\int_0^u\{\theta(v, \tilde l)-1\} d \Lambda_0(v | \tilde l)} \frac{\{\theta(\tilde{u}, \tilde l) - 1\} 1(\tilde u < u)}{S(\tilde{u} | \tilde l)} f(y | u, \tilde l) f(u |\tilde l) dy du\\
&= \int E(Y | u, \tilde l) \theta(u, \tilde l)^\delta e^{-\int_0^u\{\theta(v, \tilde l)-1\} d \Lambda_0(v | \tilde l)} \frac{\{\theta(\tilde{u}, \tilde l) - 1\} 1(\tilde u < u)}{S(\tilde{u} | \tilde l)} f(u|\tilde l) du\\
&= E \left[\left.E(Y | U', \tilde l) \theta(U', \tilde l)^{\Delta'} e^{-\int_0^{U'}\{\theta(v, \tilde l)-1\} d \Lambda_0(v | \tilde l)} \frac{\{\theta(\tilde{u}, \tilde l) - 1\} 1(\tilde u < U')}{S(\tilde{u} | \tilde l)} \right|\tilde l \right]\\
&= E_{\widetilde{\P}}\left[\left. E \left[ \left. E(Y | U', L) \theta(U', L)^{\Delta'} e^{-\int_0^{U'}\{\theta(v, L)-1\} d \Lambda_0(v | L)} \frac{\{\theta(u, L) - 1\}}{S(u | L)} 1(u < U') \right| L \right] \right|_{u=U} \right].
\end{align*}
Combining these two terms with $g(\tilde o, \P)-E_{\P}\{g(O, \P)\}$ yields the efficient influence function as
\begin{align*}
& Y \theta(U,L)^{\Delta} e^{-\int_0^U\{\theta(v,L) - 1\}d\Lambda_0(v|L)} \\
& + \left. E \left[ \left. E(Y | U', L) \theta(U', L)^{\Delta'} e^{-\int_0^{U'}\{\theta(v, L)-1\} d \Lambda_0(v | L)} \int_0^{u \wedge U'} \frac{\{\theta(v, L)-1\}}{S(v|L)} d\Lambda_0(v | L) \right| L \right] \right|_{u=U} \\
& - \left. E \left[ \left. E(Y | U', L) \theta(U', L)^{\Delta'} e^{-\int_0^{U'}\{\theta(v, L)-1\} d \Lambda_0(v | L)} \frac{\{\theta(u, L) - 1\}}{S(u | L)} 1(u < U') \right| L \right] \right|_{u=U}
-\psi(\theta),
\end{align*}
where $U'$ is an independent copy of $U$ given $L$, $\Delta' = 1(U' < \tau)$, and $S(u|l) = e^{- \Lambda_0(u|l)}$. The outermost expectation $E$ is taken with respect to $U'$ conditional on $L$. This function can also be expressed as $\phi(\theta; \Lambda_0, \mu_0) -\psi(\theta)$, where 
\begin{align*}
\phi(\theta; \Lambda_0, \mu_0) =& Y \theta(U,L)^{\Delta} e^{-\int_0^U\{\theta(v,L) - 1\}d\Lambda_0(v|L)} \\
& + \int_0^\tau \mu_0(u, L) \theta(u, L)^\delta e^{-\int_0^u\{\theta(v, L)-1\} d \Lambda_0(v | L)} \left[\int_0^{U \wedge u} \left\{\theta(v, L)-1\right\} d e^{\Lambda_0(v | L)} \right] d F_0(u|L) \\
& - \frac{\theta(U, L) - 1}{e^{-\Lambda_0(U | L)}} \int_{U+}^\tau \mu_0(u, L) \theta(u, L)^\delta e^{-\int_0^u \{\theta(v, L)-1\} d \Lambda_0(v | L)} d F_0(u|L), 
\end{align*}
where $\delta = 1(u < \tau)$ and $ F_0(u|l) = 1 -e^{- \Lambda_0(u|l)}$.

Following \cite{hines2022demystifying}, we verify that $\phi(\theta;\Lambda_0,\mu_0)$ has finite variance by showing that $E\{\phi(\theta;\Lambda_0,\mu_0)^2\} < \infty$. 
We first derive several bounds for real-valued functions of $t,u \in [0, \tau]$, $l \in \cL$, and $\delta\in\{0, 1\}$ under \assumpref{bound}. First,
\begin{align*}
|\theta(t, l)|^\delta e^{-\int_0^t \{\theta(v, l) - 1\} d \Lambda_0(v| l)} 
& \le \max\cur{1, \suptl |\theta(t, l)|} \cdot e^{\abs{\int_0^t \{\theta(v, l) - 1\} d \Lambda_0(v| l)} } \\
& \le \max\cur{1, \suptl |\theta(t, l)|} \cdot e^{\cur{\suptl |\theta(t, l)| + 1} \cdot \TV\cur{\Lambda_0(\cdot|l)} } \le C_1.
\end{align*}
The last inequality holds because, for each fixed $l \in \cL$, $\Lambda_0(0|l)=0$ and $\Lambda_0(t|l)$ is nondecreasing on $[0,\tau]$, so $\TV\cur{\Lambda_0(\cdot|l)} = \Lambda_0(\tau|l)$. Here $C_1 < \infty$ is a constant that depends only on the uniform bounds of $\theta(t,l)$ and $\Lambda_0(t|l)$ over $t \in [0,\tau]$ and $l \in \cL$.

Second,
\begin{align*}
 & \abs{ \int_0^\tau \mu_0(t, l) \left[\int_0^{u \wedge t} \left\{\theta(v, l)-1\right\} d e^{\Lambda_0(v | l)} \right] d e^{-\int_0^t \theta(v, l) d \Lambda_0(v | l)} } \\
 & \le \suptl |\mu_0(t, l)| \cdot \suptl \left| \int_0^t \{\theta(v, l) - 1\} d e^{\Lambda_0(v|l)} \right| \cdot \TV\cur{e^{-\int_0^\cdot \theta(v, l) d \Lambda_0(v | l)} } \\
 & \le \suptl |\mu_0(t, l)| \cdot \cur{\suptl |\theta(t, l)| + 1} \cdot \TV\cur{e^{\Lambda_0(\cdot|l)}} \cdot \TV\cur{e^{-\int_0^\cdot \theta(v, l) d \Lambda_0(v | l)} } \le C_2,
\end{align*}
The last inequality holds because $\TV\cur{e^{\Lambda_0(\cdot|l)}} = e^{\Lambda_0(\tau|l)} - 1$, and $\TV\cur{e^{-\int_0^\cdot \theta(v, l) d \Lambda_0(v | l)} } \leq 1$, since $t \mapsto e^{-\int_0^t \theta(v, l) d \Lambda_0(v | l)}$ is nonincreasing on $[0,\tau]$ and takes values in $[0,1]$. Here $C_2 < \infty$ is a constant that depends only on the uniform bounds of $\theta(t,l)$, $\Lambda_0(t|l)$, and $\mu_0(t,l)$ over $t \in [0,\tau]$ and $l \in \cL$.

Third,
\begin{align*}
& \abs{\frac{\theta(u,l) - 1}{e^{-\Lambda_0(u | l)}} \int_{u+}^\tau \mu_0(t,l) d e^{-\int_0^t \theta(v, l) d \Lambda_0(v | l)}} \\
& \le \cur{\suptl |\theta(t, l)| + 1 } \cdot e^{\suptl |\Lambda_0(t|l)|} \cdot \suptl |\mu_0(t, l) | \cdot \TV\cur{e^{-\int_0^\cdot \theta(v, l) d \Lambda_0(v | l)} }
\le C_3,
\end{align*}
where $C_3 < \infty$ is a constant that depends only on the uniform bounds of $\theta(t,l)$, $\Lambda_0(t|l)$, and $\mu_0(t,l)$ over $t \in [0,\tau]$ and $l \in \cL$.

Therefore, by the representation in \eqref{eq:pf_eif}, $|\phi(\theta; \Lambda_0, \mu_0)| \le C_1 |Y| + C_2 + C_3$. Hence, $E\{\phi(\theta; \Lambda_0, \mu_0)^2\} \le E\{ 2 C_1^2 Y^2 + 2(C_2 + C_3)^2 \} < \infty$, since $E(Y^2) < \infty$ under \assumpref{bound}.
\end{proof}

\subsection{Equivalence of the Efficient Influence Function Representations} \label{sec:eifeq}

We show that  \eqref{eq:org_eif} is equivalent to \eqref{eq:pf_eif} below: 
\begin{align}
\phi(\theta; \Lambda_0, \mu_0) =& Y \theta(U,L)^{\Delta} e^{-\int_0^U\{\theta(v,L) - 1\}d\Lambda_0(v|L)} \notag \\
& - \int_0^\tau \mu_0(u, L) \left[\int_0^{U \wedge u} \left\{\theta(v, L)-1\right\} d e^{\Lambda_0(v | L)} \right] d e^{-\int_0^u \theta(v, L) d \Lambda_0(v | L)} \notag \\
& + \frac{\theta(U, L) - 1}{e^{-\Lambda_0(U | L)}} \int_{U+}^\tau \mu_0(u, L) d e^{-\int_0^u \theta(v, L) d \Lambda_0(v | L)}. \label{eq:pf_eif}
\end{align}

\begin{proof}
If $u < \tau$, 
\begin{align*}
d e^{-\int_0^u \theta(v, l) d \Lambda_0(v | l)}
&= -\theta(u, l) e^{-\int_0^u \theta(v, l) d \Lambda_0(v | l)} d \Lambda_0(u | l) \\
&= -\theta(u, l) e^{-\int_0^u \{\theta(v, l)-1\} d \Lambda_0(v | l)} e^{-\Lambda_0(u | l)} d \Lambda_0(u | l) \\
&= -\theta(u, l) e^{-\int_0^u \{\theta(v, l)-1\} d \Lambda_0(v | l)} d \cur{1-e^{-\Lambda_0(u | l)}} \\
&= -\theta(u, l) e^{-\int_0^u \{\theta(v, l)-1\} d \Lambda_0(v | l)} d F_0(u|l).
\end{align*}
If $u = \tau$, 
\begin{align*}
    d e^{-\int_0^u \theta(v, l) d \Lambda_0(v | l)} \coloneqq 0 - e^{-\int_0^u \theta(v, l) d \Lambda_0(v | l)} 
    &= - e^{-\int_0^u \{\theta(v, l)-1\} d \Lambda_0(v | l)} e^{-\int_0^u d \Lambda_0(v | l)} \\
    &= - e^{-\int_0^u \{\theta(v, l)-1\} d \Lambda_0(v | l)} \cdot \cur{1-F_0(u-|l)} \\
    &= - e^{-\int_0^u \{\theta(v, l)-1\} d \Lambda_0(v | l)} d F_0(u|l).
\end{align*}
Combining these two cases gives
\begin{equation*}
    d e^{-\int_0^u \theta(v, l) d \Lambda_0(v | l)} = -\theta(u, l)^\delta e^{-\int_0^u \{\theta(v, l)-1\} d \Lambda_0(v | l)} d F_0(u|l),
\end{equation*}
where $\delta = 1(u <\tau)$. Using this fact, the equivalence follows.
\end{proof}

\subsection{Proof of \corref{eifiden}} \label{sec:pfeifidencor}

\begin{proof}
When $\theta(t,l) \equiv \theta$, we have
$$
e^{-\int_0^u\{\theta(v, L)-1\} d \Lambda_0(v | L)} 
= e^{-(\theta - 1) \Lambda_0(u | L)},
$$
and
$$
\int_0^{u \wedge U'} \frac{\{\theta(v, L)-1\}}{S(v|L)} d\Lambda_0(v | L) = \left. \frac{\theta - 1}{S(v | L)}\right|_0^{u \wedge U'} = (\theta - 1)\left\{\frac{1(u < U')}{S(u | L)} + \frac{1(u \geq U')}{S(U' | L)} - 1 \right\},
$$
which simplifies the efficient influence function in \thmref{eifiden} to
\begin{align*}
& Y \theta^\Delta e^{-(\theta - 1) \Lambda_0(U | L)} \\
& + (\theta - 1) \left. E \left\{ \left. E(Y | U', L) \theta^{\Delta'} e^{-(\theta - 2) \Lambda_0(U' | L)} 1(u \geq U') \right| L \right\} \right|_{u=U} \\
& - (\theta - 1) E \left\{ \left. E(Y | U, L) \theta^\Delta e^{-(\theta - 1) \Lambda_0(U | L)} \right| L \right\}
- \psi(\theta).
\end{align*}
This function can also be expressed as $\phi(\theta; \Lambda_0, \mu_0) -\psi(\theta)$, where 
\begin{align*}
\phi(\theta; \Lambda_0, \mu_0) = & Y \theta^\Delta e^{-(\theta - 1) \Lambda_0(U | L)} \\
& + (\theta - 1) \int_0^U \mu_0(u, L) \theta^\delta e^{-(\theta - 2)\Lambda_0(u|L)} d F_0(u|L) \\
& - (\theta - 1) \int_0^\tau \mu_0(u, L) \theta^\delta e^{-(\theta - 1)\Lambda_0(u|L)} d F_0(u|L).
\end{align*}
\end{proof}

\subsection{Proof of \thmref{eifcons}} \label{sec:pfeifcons}

We derive several bounds for real-valued functions of $t \in [0,\tau]$ and $l \in \cL$ that will be used in the proof of \thmref{eifcons}.
We will use the fact that the Riemann–Stieltjes integral $\int_a^b f(x) dg(x)$ exists 
{if both $f$ and $g$ are of bounded variation and share no common discontinuities \citep{young1936inequality}. }
This integral admits integration by parts in the form $\int_a^b f(x) dg(x) = f(b)g(b) - f(a)g(a) - \int_a^b g(x) df(x)$
\citep{apostol1974mathematical}. 

First, under \assumpref{bound}, we have the bound
\begin{align}
    & \abs{e^{\hat \Lambda(t|l)} - e^{\Lambda_0(t|l)}} 
    \leq \supt \squ{e^{\max \cur{\hat \Lambda(t|l),\ \Lambda_0(t|l)}} \abs{\hat\Lambda(t|l) - \Lambda_0(t|l)} } \notag \\
    & \lesssim e^{\max \cur{\hat \Lambda(\tau|l),\ \Lambda_0(\tau|l)}} \cdot \suplaml
    \lesssim e^{\hat \Lambda(\tau|l)} \cdot \suplaml, \label{eq:bound_elam}
\end{align}
where the first inequality holds since for any $a\leq b$, by the mean value theorem, $|e^a-e^b|=e^c |a-b| \leq e^{\max\cur{|a|,|b|}} |a-b|$ for some $c\in[a, b]$, the second follows from the monotonicity of the exponential and (estimated) cumulative hazard functions, and the last one follows from $e^{\max\cur{|a|,|b|}} \leq e^{|a|+|b|}$ and \assumpref{bound:nuisance}.

For $(t,l) \in [0,\tau]\times\cL$, define
\begin{equation}
    G(t,l;\Lambda) = \int_0^t \{\theta(v,l)-1\} de^{\Lambda(v|l)}, 
    \quad
    S(t,l;\Lambda) = e^{-\int_0^t \theta(v,l) d\Lambda(v|l)}. \label{eq:G_S_def}
\end{equation}
Under \assumpref{bound}, using integration by parts, by \eqref{eq:bound_elam},
\begin{align}
    & \abs{G(t,l;\hat\Lambda) - G(t,l;\Lambda_0)} = \abs{\int_0^{t}\{\theta(v,l) - 1\}d \cur{e^{\hat \Lambda(v|l)} - e^{\Lambda_0(v|l)} }} \notag \\
    &= \abs{\left.\{\theta(v,l) - 1\} \cur{e^{\hat \Lambda(v|l)} - e^{\Lambda_0(v|l)} } \right|_0^{t} - \int_0^{t} \cur{e^{\hat \Lambda(v|l)} - e^{\Lambda_0(v|l)} } d \{\theta(v,l) - 1\} } \notag \\
    &\leq \abs{\theta(t,l)-1} \cdot \abs{e^{\hat\Lambda(t|l)} - e^{\Lambda_0(t|l)}} + \sup_{v\in [0,t]} \abs{e^{\hat \Lambda(v|l)} - e^{\Lambda_0(v|l)}} \cdot \TV\cur{\theta(\cdot,l) - 1}, \notag \\
    &\lesssim \supt \abs{e^{\hat \Lambda(t|l)} - e^{\Lambda_0(t|l)}}
    \lesssim e^{\hat \Lambda(\tau|l)} \cdot \suplaml. \label{eq:bound_int_delam}
\end{align}
Likewise, under \assumpref{bound}, using integration by parts,
\begin{equation}
    \abs{\int_0^{t}\{\theta(v,l) - 1\}d \cur{ \hat \Lambda(v|l) - \Lambda_0(v|l) }} \lesssim \suplaml \label{eq:bound_int_dlam}.
\end{equation}
Then, under \assumpref{bound}, by the mean value theorem and \eqref{eq:bound_int_dlam},
\begin{align}
    & \abs{S(t,l;\hat\Lambda) - S(t,l;\Lambda_0) } 
    = \abs{e^{-\int_0^{t} \theta(v,l) d \hat \Lambda(v|l)} - e^{-\int_0^{t} \theta(v,l) d \Lambda_0(v|l)}} \notag \\
    & \leq e^{\max \squ{ \abs{\int_0^{t} \theta(v,l) d \hat \Lambda(v|l)},\ \abs{\int_0^{t} \theta(v,l) d \Lambda_0(v|l)} } } \cdot \abs{\int_0^{t} \theta(v,l) d \cur{ \hat \Lambda(v|l) - \Lambda_0(v|l) }} \notag \\
    & \lesssim e^{\sup_{t,l} |\theta(t,l)| \cdot \hat \Lambda(\tau|l)} \cdot \suplaml. \label{eq:bound_exp_theta}
\end{align}
Likewise, under \assumpref{bound}, we can show that
\begin{equation}
    \abs{e^{-\int_0^{t}\{\theta(v,l) - 1\}d \hat \Lambda(v|l)} - e^{-\int_0^{t}\{\theta(v,l) - 1\}d \Lambda_0(v|l)}}
    \lesssim e^{\sup_{t,l} |\theta(t,l)-1| \cdot \hat \Lambda(\tau|l)} \cdot \suplaml. \label{eq:bound_exp_theta_1}
\end{equation}

Additionally, as mentioned in the proof of \thmref{eifiden}, $\TV\cur{e^{-\int_0^\cdot \theta(v, l) d \Lambda_0(v | l)} } \leq 1$. Likewise, $\TV\cur{e^{-\int_0^\cdot \theta(v, l) d \hat \Lambda(v | l)} } \leq 1$. Consequently, $\TV\cur{e^{-\int_0^\cdot \theta(v, l) d \hat \Lambda(v | l)} - e^{-\int_0^\cdot \theta(v, l) d \Lambda_0(v | l)}} \leq \TV\cur{e^{-\int_0^\cdot \theta(v, l) d \hat \Lambda(v | l)} } + \TV\cur{e^{-\int_0^\cdot \theta(v, l) d \Lambda_0(v | l)} } \leq 2$.

\begin{proof}
We have
\begin{align}
& \left|\hat \psi(\theta) - \psi(\theta) \right| = \abs{ P_n \{\phi(\theta; \hat \Lambda, \hat\mu)\} - P \{\phi(\theta; \Lambda_0, \mu_0)\}} \notag \\
&= \Big| P_n \left[ \left\{ \phi(\theta; \hat \Lambda, \hat \mu) - \phi(\theta; \Lambda_0, \hat \mu) \right\} - \left\{ \phi(\theta; \hat \Lambda, \mu_0) - \phi(\theta; \Lambda_0, \mu_0) \right\} \right] \notag \\
&\phantom{= |} + P_n \cur{\phi(\theta; \hat \Lambda, \mu_0) - \phi(\theta; \Lambda_0, \mu_0) }
+ P_n \cur{\phi\bra{\theta; \Lambda_0, \hat\mu} - \phi\bra{\theta; \Lambda_0, \mu_0} }
+ P_n \{\phi(\theta; \Lambda_0, \mu_0)\} - P \{\phi(\theta; \Lambda_0, \mu_0)\} \Big| \notag \\
&\leq \abs{P_n \squ{ \cur{ \phi(\theta; \hat \Lambda, \hat \mu) - \phi(\theta; \Lambda_0, \hat \mu) } - \cur{ \phi(\theta; \hat \Lambda, \mu_0) - \phi(\theta; \Lambda_0, \mu_0) } } } \notag \\
&\phantom{\leq |} + \abs{P_n \cur{\phi(\theta; \hat \Lambda, \mu_0) - \phi(\theta; \Lambda_0, \mu_0) }} 
+ \abs{P_n \cur{\phi\bra{\theta; \Lambda_0, \hat\mu} - \phi\bra{\theta; \Lambda_0, \mu_0} }} 
+ o_p(1),  \label{eq:pfcons_diff}
\end{align}
where the $o_p(1)$ term follows from the Weak Law of Large Numbers. To show $\abs{\hat \psi(\theta) - \psi(\theta)} = o_p(1)$, it suffices to show that the first three terms in \eqref{eq:pfcons_diff} are both $o_p(1)$.

Consider the first term in \eqref{eq:pfcons_diff}. We have
\begin{equation*}
\abs{P_n \squ{ \cur{ \phi(\theta; \hat \Lambda, \hat \mu) - \phi(\theta; \Lambda_0, \hat \mu) } - \cur{ \phi(\theta; \hat \Lambda, \mu_0) - \phi(\theta; \Lambda_0, \mu_0) } } } \leq A_{11} + A_{12} + A_{13} + A_{14} + A_{15} + A_{16},
\end{equation*}
where
\begin{align*}
A_{11} =& \frac{1}{n} \sum_{i=1}^n \bigg| \int_0^\tau \{\hat \mu (u, L_i) - \mu_0 (u, L_i)\} \\
&\phantom{\frac{1}{n} \sum_{i=1}^n \bigg|} \cdot \int_0^{U_i \wedge u} \left\{\theta(v, L_i)-1\right\} d e^{\Lambda_0(v | L_i)} d\left\{e^{-\int_0^u \theta(v, L_i) d \hat\Lambda(v | L_i)} - e^{-\int_0^u \theta(v, L_i) d \Lambda_0(v | L_i)}\right\} \bigg| , \\
A_{12} =& \frac{1}{n} \sum_{i=1}^n \bigg| \int_0^\tau \{\hat \mu (u, L_i) - \mu_0 (u, L_i)\}
\int_0^{U_i \wedge u} \left\{\theta(v, L_i)-1\right\} d \left\{e^{\hat \Lambda(v | L_i)} - e^{\Lambda_0(v | L_i)} \right\} d e^{-\int_0^u \theta(v, L_i) d \Lambda_0(v | L_i)} \bigg| , \\
A_{13} =& \frac{1}{n} \sum_{i=1}^n \bigg| \int_0^\tau \{\hat \mu (u, L_i) - \mu_0 (u, L_i)\}  \\  
&\phantom{\frac{1}{n} \sum_{i=1}^n \bigg|} \cdot \int_0^{U_i \wedge u} \left\{\theta(v, L_i)-1\right\} d \left\{e^{\hat \Lambda(v | L_i)} - e^{\Lambda_0(v | L_i)} \right\} d\left\{e^{-\int_0^u \theta(v, L_i) d \hat\Lambda(v | L_i)} - e^{-\int_0^u \theta(v, L_i) d \Lambda_0(v | L_i)}\right\} \bigg| , \\
A_{14} =& \frac{1}{n} \sum_{i=1}^n \bigg| \{\theta(U_i, L_i) - 1\} \int_{U_i+}^\tau \{\hat \mu (u, L_i) - \mu_0 (u, L_i)\}
e^{\Lambda_0(U_i | L_i)} d\left\{e^{-\int_0^u \theta(v, L_i) d \hat\Lambda(v | L_i)} - e^{-\int_0^u \theta(v, L_i) d \Lambda_0(v | L_i)}\right\} \bigg|, \\
A_{15} =& \frac{1}{n} \sum_{i=1}^n \abs{\{\theta(U_i, L_i) - 1\} \int_{U_i+}^\tau \{\hat \mu (u, L_i) - \mu_0 (u, L_i)\}  \left\{e^{\hat \Lambda(U_i | L_i)} - e^{\Lambda_0(U_i | L_i)} \right\} d e^{-\int_0^u \theta(v, L_i) d \Lambda_0(v | L_i)} } , \\
A_{16} =& \frac{1}{n} \sum_{i=1}^n \bigg| \{\theta(U_i, L_i) - 1\} \int_{U_i+}^\tau \{\hat \mu (u, L_i) - \mu_0 (u, L_i)\}  \\
&\phantom{\frac{1}{n} \sum_{i=1}^n \bigg|} \cdot \left\{e^{\hat \Lambda(U_i | L_i)} - e^{\Lambda_0(U_i | L_i)} \right\} d\left\{e^{-\int_0^u \theta(v, L_i) d \hat\Lambda(v | L_i)} - e^{-\int_0^u \theta(v, L_i) d \Lambda_0(v | L_i)}\right\} \bigg|.
\end{align*}
For $A_{11}$, under Assumptions~\ref{assump:bound} and \ref{assump:nuiest},
\begin{align*}
    E(A_{11}) =& E \bigg\{ \bigg| \int_0^\tau \{\hat \mu (u, L_i) - \mu_0 (u, L_i)\}  \\
    &\phantom{E \bigg\{ \bigg|} \cdot \int_0^{U_i \wedge u} \left\{\theta(v, L_i)-1\right\} d e^{\Lambda_0(v | L_i)} d\cur{e^{-\int_0^u \theta(v, L_i) d \hat\Lambda(v | L_i)} - e^{-\int_0^u \theta(v, L_i) d \Lambda_0(v | L_i)}} \bigg| \bigg\} \\
    \lesssim & E \cur{\supmui \cdot \TV\cur{e^{-\int_0^u \theta(v, L_i) d \hat\Lambda(v | L_i)} - e^{-\int_0^u \theta(v, L_i) d \Lambda_0(v | L_i)}}} \\
    \lesssim & E \cur{\supmui} \lesssim E \cur{\supmui^2}^{1/2} \lesssim \nsupmu = o(1).
\end{align*}
For $A_{12}$, under Assumptions~\ref{assump:bound} and \ref{assump:nuiest}, by \eqref{eq:bound_int_delam},
\begin{align*}
    E(A_{12}) =& E \bigg\{ \bigg| \int_0^\tau \{\hat \mu (u, L_i) - \mu_0 (u, L_i)\}
    \int_0^{U_i \wedge u} \left\{\theta(v, L_i)-1\right\} d \left\{e^{\hat \Lambda(v | L_i)} - e^{\Lambda_0(v | L_i)} \right\} d e^{-\int_0^u \theta(v, L_i) d \Lambda_0(v | L_i)} \bigg| \bigg\} \\
    \lesssim & E \Bigg\{ \supmui \cdot e^{\hat \Lambda(\tau|L_i)} \cdot \suplami
    \cdot \TV\cur{e^{-\int_0^u \theta(v, L_i) d \Lambda_0(v | L_i)}} \Bigg\} \\
    \lesssim & E \cur{\supmui \cdot \suplami} \\
    \lesssim & E \cur{\supmui^2}^{1/2} E \cur{\suplami^2}^{1/2} \\
    \lesssim & \nsupmu \cdot \nsuplam = o(1).
\end{align*}
Similarly, under Assumptions~\ref{assump:bound} and \ref{assump:nuiest}, by \eqref{eq:bound_int_delam} and \eqref{eq:bound_elam}, we can show that
\begin{align*}
    E(A_{13}) &\lesssim \nsupmu \cdot \nsuplam = o(1), \quad
    E(A_{14}) \lesssim \nsupmu = o(1), \\
    E(A_{15}) &\lesssim \nsupmu \cdot \nsuplam = o(1), \quad
    E(A_{16}) \lesssim \nsupmu \cdot \nsuplam = o(1).
\end{align*}
Therefore,
\begin{align*}
    &E\squ{ \abs{P_n \squ{ \cur{ \phi(\theta; \hat \Lambda, \hat \mu) - \phi(\theta; \Lambda_0, \hat \mu) } - \cur{ \phi(\theta; \hat \Lambda, \mu_0) - \phi(\theta; \Lambda_0, \mu_0) } } } } \\
    &\leq E\bra{A_{11} + A_{12} + A_{13} + A_{14} + A_{15} + A_{16}} = o(1).
\end{align*}
By Markov's inequality, 
\begin{equation*}
    \abs{P_n \squ{ \cur{ \phi(\theta; \hat \Lambda, \hat \mu) - \phi(\theta; \Lambda_0, \hat \mu) } - \cur{ \phi(\theta; \hat \Lambda, \mu_0) - \phi(\theta; \Lambda_0, \mu_0) } } } = o_p(1).
\end{equation*}

Consider the second term in \eqref{eq:pfcons_diff}. We have
\begin{equation*}
\abs{P_n \{\phi(\theta; \hat \Lambda, \mu_0)\} - P_n \{\phi(\theta; \Lambda_0, \mu_0)\}} \leq A_{21} + A_{22} + A_{23} + A_{24} + A_{25},
\end{equation*}
where
\begin{align*}
A_{21} = & \frac{1}{n} \sum_{i=1}^n \abs{Y_i \theta(U_i,L_i)^{\Delta_i}} \abs{e^{-\int_0^{U_i}\{\theta(v,L_i) - 1\}d \hat \Lambda(v|L_i)} - e^{-\int_0^{U_i}\{\theta(v,L_i) - 1\}d \Lambda_0(v|L_i)}}, \\
A_{22} = & \frac{1}{n} \sum_{i=1}^n \int_0^\tau \abs{\mu_0 (u, L_i) } \abs{\int_0^{U_i \wedge u} \{\theta(v, L_i)-1\}d \cur{e^{\hat \Lambda(v|L_i)} - e^{\Lambda_0(v|L_i)} }} d e^{-\int_0^u \theta(v, L) d \hat \Lambda(v | L)}, \\
A_{23} = & \frac{1}{n} \sum_{i=1}^n \abs{\int_0^\tau \mu_0 (u, L_i)  \int_0^{U_i \wedge u} \{\theta(v, L_i)-1\}d e^{\Lambda_0(v|L_i)} d\left\{e^{-\int_0^u \theta(v, L_i) d \hat\Lambda(v | L_i)} - e^{-\int_0^u \theta(v, L_i) d \Lambda_0(v | L_i)}\right\}}, \\
A_{24} = & \frac{1}{n} \sum_{i=1}^n \abs{\theta(U_i, L_i) - 1} \int_{U_i+}^\tau \abs{\mu_0 (u, L_i) } \abs{e^{\hat \Lambda(U_i|L_i)} - e^{\Lambda_0(U_i|L_i)}} d e^{-\int_0^u \theta(v, L) d \hat \Lambda(v | L)}, \\
A_{25} = & \frac{1}{n} \sum_{i=1}^n \abs{\theta(U_i, L_i) - 1} \abs{\int_{U_i+}^\tau \mu_0 (u, L_i)  e^{\Lambda_0(U_i|L_i)} d\left\{e^{-\int_0^u \theta(v, L_i) d \hat\Lambda(v | L_i)} - e^{-\int_0^u \theta(v, L_i) d \Lambda_0(v | L_i)}\right\}}.
\end{align*}
For $A_{21}$, under Assumptions~\ref{assump:bound} and \ref{assump:nuiest}, by \eqref{eq:bound_exp_theta_1}, we have
\begin{align*}
    E(A_{21}) 
    & = E\squ{\abs{Y_i \theta(U_i,L_i)^{\Delta_i}} \abs{e^{-\int_0^{U_i}\{\theta(v,L_i) - 1\}d \hat \Lambda(v|L_i)} - e^{-\int_0^{U_i}\{\theta(v,L_i) - 1\}d \Lambda_0(v|L_i)}}} \\
    & \lesssim E \cur{Y_i^2 \theta(U_i,L_i)^{2\Delta_i}}^{1/2} E\cur{e^{\sup_{t,l} |\theta(t,l)-1| \cdot \hat \Lambda(\tau|L_i)} \cdot 
    \suplami^2}^{1/2} \\
    & \lesssim \nsuplam = o(1).
\end{align*}
For $A_{22}$, under Assumptions~\ref{assump:bound} and \ref{assump:nuiest}, by \eqref{eq:bound_int_delam}, we have
\begin{align*}
    E(A_{22}) &= E \squ{\int_0^\tau \abs{\mu_0 (u, L_i) } \abs{\int_0^{U_i \wedge u} \{\theta(v, L_i)-1\}d \left\{e^{\hat \Lambda(v|L_i)} - e^{\Lambda_0(v|L_i)} \right\}} d e^{-\int_0^u \theta(v, L) d \hat \Lambda(v | L)}} \\
    &\lesssim E \squ{e^{\hat \Lambda(\tau|L_i)} \cdot \suplami \cdot \TV\cur{ e^{-\int_0^\cdot \theta(v, l) d \hat \Lambda(v | l)} }} \\
    & \lesssim E \cur{\suplami^2}^{1/2} \lesssim \nsuplam = o(1).
\end{align*}
For $A_{23}$, under Assumptions~\ref{assump:bound} and \ref{assump:nuiest}, by \eqref{eq:bound_exp_theta} and integration by parts, we have
\begin{align*}
    E(A_{23}) &\lesssim E \squ{ e^{\sup_{t,l} |\theta(t,l)| \cdot \hat \Lambda(\tau|L_i) } \cdot \suplami \cdot \squ{ \supt \abs{\mu_0(t,L_i)} + \TV\cur{\mu_0(\cdot,L_i)}}} \\
    & \lesssim E \cur{\suplami^2}^{1/2} \lesssim \nsuplam = o(1).
\end{align*}
Similarly, under Assumptions~\ref{assump:bound} and \ref{assump:nuiest}, by \eqref{eq:bound_elam} and \eqref{eq:bound_exp_theta}, we can show that
$$
E(A_{24}) + E(A_{25}) \lesssim \nsuplam = o(1).
$$
Therefore, 
\begin{equation*}
    E \squ{\abs{P_n \{\phi(\theta; \hat \Lambda, \mu_0)\} - P_n \{\phi(\theta; \Lambda_0, \mu_0)\}}}
    \leq E(_{21} + A_{22} + A_{23} + A_{24} + A_{25}) = o(1).
\end{equation*}
By Markov's inequality,
\begin{equation*}
    \abs{P_n \{\phi(\theta; \hat \Lambda, \mu_0)\} - P_n \{\phi(\theta; \Lambda_0, \mu_0)\}} = o_p(1).
\end{equation*}

Consider the third term in \eqref{eq:pfcons_diff}. Under Assumptions~\ref{assump:bound} and \ref{assump:nuiest}, similarly we have
\begin{align*}
& E\squ{\abs{P_n \cur{\phi(\theta; \Lambda_0, \hat\mu)\} - P_n \{\phi(\theta; \Lambda_0, \mu_0)}}} \\
&\leq E\squ{\abs{\int_0^\tau \{\hat \mu (u, L_i) - \mu_0 (u, L_i) \}  \left[ \int_0^{U_i \wedge u} \{\theta(v, L_i)-1\} d e^{\Lambda_0(v|L_i)} \right] d e^{-\int_0^u \theta(v, L_i) d \Lambda_0(v|L_i)}}} \\
&\phantom{\leq} + E \squ{\abs{ \frac{\theta(U_i, L_i) - 1}{e^{-\Lambda_0(U_i | L_i)}} \int_{U_i+}^\tau \{\hat \mu (u, L_i) - \mu_0 (u, L_i) \}  d e^{-\int_0^u \theta(v, L_i) d \Lambda_0(v|L_i)}} } \\
&\lesssim E \cur{\supmui^2}^{1/2} \lesssim \nsupmu = o(1).
\end{align*}
Therefore, by Markov's inequality,
\begin{equation*}
    \abs{P_n \cur{\phi(\theta; \Lambda_0, \hat\mu)\} - P_n \{\phi(\theta; \Lambda_0, \mu_0)}} = o_p(1).
\end{equation*}
This completes the proof.
\end{proof}

\subsection{Proof of \thmref{cf}} \label{sec:cf_proof}

The following is an introduction and a lemma related to the Gateaux derivative. Following \cite{bickel1993efficient}, let $f$ be a function from a linear space $X$ to another linear space $Y$. The Gateaux derivative of $f$ at $x \in X$ in the direction $h \in X$ is defined as $$Df(x)[h] \coloneqq \lim_{t \to 0} \{f(x + t h) - f(x)\} / t,$$ whenever the limit exists for all sufficiently small $t$ such that $x + t h \in X$. 

We consider the space $X \coloneqq \{ g(t,l): [0,\tau]\times \cL \to \bbR \}$. 
We can show that the Gateaux derivatives of $e^{\Lambda(t|l)}$, $G(t,l;\Lambda)$ and $S(t,l;\Lambda)$ defined in \eqref{eq:G_S_def}, at $\Lambda \in X$ in the direction $h \in X$ are 
\begin{align}
D \{e^{\Lambda(t|l)}\}[h] &= e^{\Lambda(t|l)} h(t,l), \label{eq:Dexp}\\
D G(t,l;\Lambda)[h] &= \int_0^t \{\theta(v,l)-1\} d \cur{e^{\Lambda(v|l)} h(v,l)} , \label{eq:DG} \\
D S(t,l;\Lambda)[h] &= - e^{-\int_0^t \theta(v,l) d\Lambda(v|l)} \int_0^t \theta(v,l) d h(v,l)  
= -S(t,l;\Lambda) \int_0^t \theta(v,l) d h(v,l). \label{eq:DS} 
\end{align}

\begin{lem} \label{lem:sup_tv_norm_bound}
Consider the cumulative hazard function $\Lambda_0(t|l)$ and its estimate  $\hat\Lambda(t|l)$. Under Assumptions~\ref{assump:bound} and \ref{assump:nuiest}, 
the following  hold for any $l\in\cL$:
\begin{align*}
\supt \abs{ G(t,l; \hat\Lambda) - G(t,l; \Lambda_0) }
& \lesssim e^{\hat \Lambda(\tau|l)} \cdot \suplaml, \\
\supt \abs{ S(t,l; \hat\Lambda) - S(t,l; \Lambda_0)}
& \lesssim e^{\sup_{t,l} |\theta(t,l)| \cdot \hat \Lambda(\tau|l) } \cdot \suplaml, \\
\supt \abs{DG(t,l; \hat\Lambda)[h] - DG(t,l; \Lambda_0)[h]}
& \lesssim e^{\hat \Lambda(\tau|l)}  \cdot \suplaml \cdot \supt |h(t,l)|, \\
\supt \abs{DS(\cdot,l; \hat\Lambda)[h] - DS(\cdot,l; \Lambda_0)[h]} 
& \lesssim e^{\sup_{t,l} |\theta(t,l)| \cdot \hat \Lambda(\tau|l)} \cdot \suplaml \cdot \supt |h(t,l)|,
\end{align*}
where the implicit constants are determined by the bounds in Assumption~\ref{assump:bound}.
\end{lem}
The proof of \lemref{sup_tv_norm_bound} is given in \secref{pf_sup_tv_norm_bound}.
\lemref{cond_exp} below uses the Gateaux derivatives in its proof. 
Recall that the efficient influence function $\phi(\theta;\Lambda_0, \mu_0)$   defined in \eqref{eq:pf_eif} is a function of $O=(Y,U,L)$. Let $\phi(\theta; \hat \Lambda, \hat \mu)$ denote its plug-in version, where the nuisance estimators are obtained from a sample $O'$ that is independent of $O$. 
\begin{lem} \label{lem:cond_exp}
Under Assumptions~\ref{assump:bound} and \ref{assump:nuiest}, 
\begin{eqnarray*}
    \abs{E\{\phi(\theta; \hat \Lambda, \hat \mu) - \phi(\theta; \Lambda_0, \mu_0) | O'\}} 
    &\lesssim& 
    \dgsupnmu \cdot \dgsupnlam + \dgsupnlam^2 + \|\hat\Lambda-\Lambda_0\|_{\dagger,\sup,4}^2.
\end{eqnarray*}
where the implicit constant in the upper bound depends on $\theta$ only through $\suptl \theta(t,l)$. In addition, if $Y$ is bounded, the bound reduces to $\dgsupnmu \cdot \dgsupnlam + \dgsupnlam^2$.
\end{lem}
The proof of \lemref{cond_exp} is given in \secref{pf_cond_exp}.

\bigskip
Finally for completeness we state several definitions used in the proof below. 
Following \cite{van1996weak}, let $(\cF, \|\cdot\|)$ 
be a class of measurable real-valued functions $f:\cX \to \bbR$, where $\cX$ is the sample space and $\|\cdot\|$ denotes the $L_2(P)$-norm with respect to the probability measure $P$ on $\cX$. An \emph{envelope function} of $\cF$ is any function $x \mapsto F(x)$ such that $|f(x)| \leq F(x)$ for all $x \in \cX$ and $f \in \cF$. The minimal envelope function is $x \mapsto \sup_{f\in\cF} |f(x)|$. 

A \emph{bracket} $[l,u]$ is the set of all functions $f$ satisfying $l(x) \le f(x) \le u(x)$ for all $x \in \cX$. 
An $\epsilon$-$L_2(P)$-bracket is a bracket $[l, u]$ with $\|u - l\| < \epsilon$. 
We write $N_{[]}(\epsilon, \cF, L_2(P))$ for the $L_2(P)$-\emph{bracketing number} of $\cF$, that is, the minimum number of $\epsilon$-$L_2(P)$-brackets needed to cover $\cF$. 
For a class $\cF$ with a square integrable envelope function $F$, the $L_2(P)$-\emph{bracketing integral} of $\cF$ is defined as
$J_{[]}(\cF) = \int_0^1 \sqrt{1 + \log N_{[]}(\epsilon \|F\|, \cF, L_2(P))} d\epsilon$.

\begin{proof}
Define the processes
\begin{align}
    \widehat{\Psi}_n(\theta) &= \sqrt{n}\{\hat{\psi}_{\cf}(\theta) - \psi(\theta)\}/\hat{\sigma}_{\cf}(\theta), \notag \\
    \widetilde{\Psi}_n(\theta) &= \sqrt{n}\{\hat{\psi}_{\cf}(\theta) - \psi(\theta)\}/\sigma(\theta), \notag \\
    \Psi_n(\theta) &= G_n \{\phi(\theta;\Lambda_0, \mu_0) / \sigma(\theta)\} \equiv G_n\{\tilde{\phi}(\theta;\Lambda_0, \mu_0)\}, \label{eq:Psi_theta}
\end{align}
where $G_n = \sqrt{n}(P_n - P)$ is the empirical process on the full sample. Recall that $P_n$ denotes the empirical average over all units, and $P$ denotes the expectation with respect to the observed data distribution.

Let $\|f\|_{\Theta} = \sup_{\theta \in \Theta} |f(\theta)|$ denote the supremum norm of a function $f$ over $\Theta$.  
In this proof, we show that
\begin{equation}
    \Psi_n(\cdot) \rightsquigarrow G(\cdot) \text{ in } \ell^{\infty}(\Theta), \label{eq:cf_const_donsker}
\end{equation}
and
\begin{equation}
    \nsuptheta{ \widehat{\Psi}_n - \Psi_n } = o_p(1), \label{eq:cf_const_sup} 
\end{equation}
which together completes the proof of \thmref{cf} by Slutsky’s theorem.  


To establish \eqref{eq:cf_const_donsker}, it suffices to show that the class $\cF_0 \coloneqq \{\phi(\theta; \Lambda_0, \mu_0)/\sigma(\theta): \theta \in \Theta\}$ is Donsker. 
By \assumpref{cf:sigma} and Corollary 2.10.15 of \cite{van1996weak},
$\cF_0$ is Donsker if the class $\cF_{\Lambda_0, \mu_0} \coloneqq \{\phi(\theta; \Lambda_0, \mu_0): \theta \in \Theta\}$ is Donsker with a square integrable envelope function. 
Under \assumpref{cf:donsker}, \eqref{eq:cf_const_donsker} follows by verifying that $\cF_{\Lambda_0, \mu_0}$ has a square integrable envelope function $F$.

Take the envelope of the class $\cF_{\Lambda_0, \mu_0}$ to be $F = \suptheta |\phi(\theta; \Lambda_0, \mu_0)|$. 
It is straightforward to verify that under the conditions of \thmref{cf}, $F$ is square integrable. This completes the proof of \eqref{eq:cf_const_donsker}. 

Consider \eqref{eq:cf_const_sup}. 
First we show that $\|\hat{\sigma}_{\cf}/\sigma - 1\|_{\Theta} = o_p(1)$ under Assumptions~\ref{assump:bound} - \ref{assump:cf}. Similar to the proof of \thmref{eifcons}, one can show that $\suptheta \abs{\hat\sigma_{\cf}(\theta) - \sigma(\theta)} = o_p(1)$.
Therefore, under \assumpref{cf:sigma}, $\|\hat{\sigma}_{\cf}/\sigma - 1\|_{\Theta} = \suptheta\abs{\cur{\hat\sigma_{\cf}(\theta) - \sigma(\theta)} / \sigma(\theta)} \leq  \suptheta \abs{\hat\sigma_{\cf}(\theta) - \sigma(\theta)} \cdot \suptheta \abs{1/\sigma(\theta)} = o_p(1)$.

Similarly to the proof in \cite{kennedy2019nonparametric}, we have
\begin{align}
    \nsuptheta{ \widehat{\Psi}_n - \Psi_n} &= \nsuptheta{(\widetilde{\Psi}_n - \Psi_n)(\sigma/\hat{\sigma}) + \Psi_n(\sigma - \hat{\sigma})/\hat{\sigma}}  \notag \\
    &\le \nsuptheta{\widetilde{\Psi}_n - \Psi_n} \nsuptheta{\sigma/\hat{\sigma}} + \nsuptheta{\sigma/\hat{\sigma} - 1} \nsuptheta{\Psi_n}  \notag \\
    &\lesssim \nsuptheta{\widetilde{\Psi}_n - \Psi_n} + o_p(1), \label{eq:bridge}
\end{align}
where the last inequality follows from the fact that $\nsuptheta{\hat{\sigma}/\sigma - 1} = o_p(1)$ and, by Theorem 2.14.17 of \cite{van1996weak}, $\nsuptheta{ \Psi_n } = O_p(1)$. 
To show \eqref{eq:cf_const_sup}, it remains to show that $\|\widetilde{\Psi}_n - \Psi_n\|_{\Theta} = o_p(1)$. 

Without loss of generality, assume $n = N\cdot K$, where $N$ is the sample size of each fold $k = 1, \dots, K$ in the cross-fitting algorithm in \secref{est}. 
Define the empirical process for group $k$ by $G_n^k = \sqrt{N}(P_n^k - P^k)$, where $P_n^k$ is the empirical average over units in fold-$k$ and $P^k$ denotes the expectation with respect to the in-fold-$k$ data distribution conditional on the out-of-fold-$k$ data. 
Let $O_k'$ denote the out-of-fold-$k$ data used to construct the nuisance estimators $\hat\Lambda_{-k}$ and $\hat\mu_{-k}$. 
Then
\begin{align*}
    \widetilde{\Psi}_n(\theta) - \Psi_n(\theta) &= \frac{\hat{\psi}(\theta) - \psi(\theta)}{\sigma(\theta)/\sqrt{n}} - G_n\{\tilde{\phi}(\theta; \Lambda_0, \mu_0)\} \\
    &= \frac{\sqrt{n}}{\sigma(\theta)} \frac{1}{K} \sum_{k=1}^K \left[ P_n^k\{\phi(\theta; \hat\Lambda_{-k}, \hat\mu_{-k})\} - \psi(\theta) - (P_n - P)\phi(\theta; \Lambda_0, \mu_0) \right] \\
    &= \frac{\sqrt{n}}{K\sigma(\theta)} \sum_{k=1}^K \left[ \frac{1}{\sqrt{N}} G_n^k \{\phi(\theta; \hat\Lambda_{-k}, \hat\mu_{-k}) - \phi(\theta; \Lambda_0, \mu_0)\} + P^k \{\phi(\theta; \hat\Lambda_{-k}, \hat\mu_{-k}) - \phi(\theta; \Lambda_0, \mu_0)\} \right] \\
    &= B_{n,1}(\theta) + B_{n,2}(\theta), 
\end{align*}
where the third equality follows by rearranging terms and noting that $\psi(\theta) = P\{\phi(\theta; \Lambda_0, \mu_0)\} =P^k \{\phi(\theta; \Lambda_0, \mu_0)\} $
and $\sum_k P_n^k\{\phi(\theta; \Lambda_0, \mu_0)\} = 
\sum_k P_n\{\phi(\theta; \Lambda_0, \mu_0)\}$, and
\begin{align}
    B_{n,1}(\theta) &= \frac{1}{\sqrt K\sigma(\theta)} \sum_{k=1}^K G_n^k \{\phi(\theta; \hat\Lambda_{-k}, \hat\mu_{-k}) - \phi(\theta; \Lambda_0, \mu_0)\} , \label{eq:B_n1} \\
    B_{n,2}(\theta) &= \frac{\sqrt{n}}{K\sigma(\theta)} \sum_{k=1}^K P^k \{\phi(\theta; \hat\Lambda_{-k}, \hat\mu_{-k}) - \phi(\theta; \Lambda_0, \mu_0)\}. \label{eq:B_n2}
\end{align}
Now, showing that $\|B_{n,1}(\theta)\|_{\Theta}$ and $\|B_{n,2}(\theta)\|_{\Theta}$ are both $o_p(1)$ completes the proof. 

For $B_{n,1}(\theta)$, by the triangle inequality and since $K$ is fixed and independent of the sample size $n$, we have
\begin{align*}
    \|B_{n,1}(\theta)\|_{\Theta} &= \sup_{\theta \in \Theta} \left| \frac{1}{\sqrt{K}\sigma(\theta)} \sum_{k=1}^K G_n^k \{\phi(\theta; \hat\Lambda_{-k}, \hat\mu_{-k}) - \phi(\theta; \Lambda_0, \mu_0)\} \right|  \\
    &\lesssim \max_k \suptheta \abs{G_n^k \{\phi(\theta; \hat\Lambda_{-k}, \hat\mu_{-k}) - \phi(\theta; \Lambda_0, \mu_0)\} } ,
\end{align*}

By Section 4.3 of \citet{hines2022demystifying}, a direct application of Chebyshev's inequality implies that $\abs{G_n^k \{\phi(\theta; \hat\Lambda_{-k}, \hat\mu_{-k}) - \phi(\theta; \Lambda_0, \mu_0)\} } =o_p(1)$ if $E\cur{\left. \abs{\phi(\theta; \hat\Lambda_{-k}, \hat\mu_{-k}) - \phi(\theta; \Lambda_0, \mu_0)}^2\right|O_k'} = o_p(1)$. Following the same bounding arguments used in the proof of \thmref{eifcons}, we can verify that
\begin{equation*}
    \abs{\phi(\theta; \hat\Lambda_{-k}, \hat\mu_{-k}) - \phi(\theta; \Lambda_0, \mu_0)} 
    \lesssim 
    \supmuk + (|Y|+1)\suplamk,
\end{equation*}
where the implicit constant depends on uniform bounds for $\theta$, $\mu_0$, $\Lambda_0$, and $\hat\Lambda$.
Therefore, under Assumptions~\ref{assump:bound} - \ref{assump:cf},
\begin{align*}
    & E\cur{\left. \abs{\phi(\theta; \hat\Lambda_{-k}, \hat\mu_{-k}) - \phi(\theta; \Lambda_0, \mu_0)}^2 \right| O_k'} \\
    & \lesssim E \Bigg\{\supmuk^2 + (|Y|+1)^2 \suplamk^2 \\
    & \phantom{\lesssim E \Bigg\{} \left.+ 2 (|Y|+1) \supmuk \suplamk \right| O_k' \Bigg\} \\
    & \lesssim \dgsupnmu^2 + \cur{E\bra{Y^4}}^{1/2} \cdot \|\hat\Lambda-\Lambda_0\|_{\dagger,\sup,4}^2 \\
    & \phantom{\lesssim \|} + \cur{E\bra{Y^2}}^{1/2} \cdot \dgsupnmu + \dgsupnmu \cdot \dgsupnlam = o_p(1).
\end{align*}
This concludes that $\|B_{n,1}(\theta)\|_{\Theta} = o_p(1)$.

For $B_{n,2}(\theta)$, by \lemref{cond_exp}, for any $k \in \{1,\dots,K\}$,
\begin{align}
    & \abs{P^k \{\phi(\theta; \hat\Lambda_{-k}, \hat\mu_{-k}) - \phi(\theta; \Lambda_0, \mu_0)\} }
    \equiv \abs{E \cur{ \left.\phi(\theta; \hat\Lambda_{-k}, \hat\mu_{-k}) - \phi(\theta; \Lambda_0, \mu_0) \right| O_k'} } \\
    & \lesssim \dgsupnmuk \cdot \dgsupnlamk + \dgsupnlamk^2 + \|\hat\Lambda_{-k}-\Lambda_0\|_{\dagger,\sup,4}^2,
\end{align}
where the implicit constant in the upper bound depends on $\theta$ only through $\suptl \theta(t,l)$. If $Y$ is bounded, the bound reduces to $\dgsupnmuk \cdot \dgsupnlamk + \dgsupnlamk^2$.

Then, under the conditions of \thmref{cf},
$$
\sup_{\theta \in \Theta} \abs{P^k \{\phi(\theta; \hat\Lambda_{-k}, \hat\mu_{-k}) - \phi(\theta; \Lambda_0, \mu_0)\}} = \opr, \text{ for any } k \in \{1,\dots,K\}.
$$
Finally, 
\begin{align*}
    \|B_{n,2}(\theta)\|_{\Theta} &= \sup_{\theta \in \Theta} \left| \frac{\sqrt{n}}{K\sigma(\theta)} \sum_{k=1}^K P^k \{\phi(\theta; \hat\Lambda_{-k}, \hat\mu_{-k}) - \phi(\theta; \Lambda_0, \mu_0)\} \right|  \\
    &\lesssim \sqrt n  \max_k \sup_{\theta \in \Theta} \abs{P^k \{\phi(\theta; \hat\Lambda_{-k}, \hat\mu_{-k}) - \phi(\theta; \Lambda_0, \mu_0)\}} = o_p(1),
\end{align*}
which concludes the proof.
\end{proof}


\subsection{Verification of \assumpref{cf} for Special Cases} \label{sec:verify}

We verify that Assumptions \ref{assump:cf:sigma} and \ref{assump:cf:donsker} hold when $\theta(t, l) \equiv \theta \in \cD = [c, C]$, where $0 < c \leq C < \infty$. The verification can be extended to a more general class of functions $\theta(t,l;\gamma)$ parameterized by a vector of parameters $\gamma \in \Gamma \subseteq \bbR^d$, where $1 \leq d < \infty$ and $\Gamma$ is compact, provided that
$\theta(t,l;\gamma)$ is continuous in $\gamma$ and its partial derivative with respect to each component of $\gamma$ exists and is uniformly bounded over $(t,l)$.

\begin{proof}
First consider the case where $\theta(t, l) \equiv \theta \in \cD = [c, C]$ with $0 < c \leq C < \infty$.
\assumpref{cf:sigma} holds because $\sigma(\theta)$ is positive and continuous (because $\phi(\theta; \Lambda_0, \mu_0)$ is continuous in $\theta$) on the compact interval $\cD$, and hence $\supth \abs{1/\sigma(\theta)} < \infty$. 

Theorem 2.5.6 of \cite{van1996weak} states that a class of measurable functions with a finite $L_2(P)$-bracketing integral is Donsker. \assumpref{cf:donsker} follows by verifying that $\cF_{\Lambda_0, \mu_0} \coloneqq \{\phi(\theta; \Lambda_0, \mu_0): \theta \in \cD\}$ has a square integrable envelope function $F$ and $\int_0^1 \sqrt{1 + \log N_{[]}(\epsilon \|F\|, \cF_{\Lambda_0, \mu_0}, L_2(P))} d\epsilon < \infty$. 
The former was shown in the proof of \thmref{cf} in \secref{cf_proof}. The latter holds since
according to Theorem 2.7.17 in \cite{van1996weak}, the $L_2(P)$-bracketing number $N_{[]}(\epsilon \|F\|, \cF_{\Lambda_0, \mu_0}, L_2(P))$ is bounded by the covering number $N (\epsilon/2, \cD, |\cdot|)$, 
which is the minimum number of balls of radius $\epsilon/2$ needed to cover the set $\cD$ under the Euclidean metric $|\cdot|$.
Consequently,
\begin{align}
    \int_0^1 \sqrt{1 + \log N_{[]}(\epsilon \|F\|, \cF_{\Lambda_0, \mu_0}, L_2(P))} d\epsilon 
    & \leq \int_0^1 \sqrt{1 + \log N \left(\frac{\epsilon}{2}, \cD, |\cdot| \right)} d\epsilon \notag \\
    & \lesssim \int_0^1 \sqrt{1 + \log \frac{1}{\epsilon}} d\epsilon < \infty, \label{eq:finite_bra_int}
\end{align}
which completes the verification of \assumpref{cf:donsker}.

The same arguments also extend to the general parametrized class described above.
\end{proof}

\subsection{Proof of \thmref{ui_const}} \label{sec:ui_const_proof}

\begin{proof}
Without loss of generality, assume $n = N\cdot K$, where $N$ is the sample size of each fold $k = 1, \dots, K$ in the cross-fitting algorithm in \secref{est}. We adopt the same notation as in the proof of \thmref{cf} in \secref{cf_proof}.

Denote by $\widehat\Psi_n^*(\theta)$ the multiplier process defined in \eqref{eq:multi_pro} of the main paper. Then
\begin{align}
    \widehat\Psi_n^*(\theta) &= \frac{1}{\sqrt{n}} \sum_{k=1}^K \sum_{i\in \cI_k} \xi_i \cur{\frac{\phi_i(\theta; \hat \Lambda_{-k}, \hat \mu_{-k}) - \hat{\psi}_{\cf}(\theta)}{\hat \sigma_{\cf}(\theta)}} 
    \equiv \frac{\sqrt n}{K} \sum_k P_n^k \squ{ \xi \cur{\frac{\phi(\theta; \hat \Lambda_{-k}, \hat \mu_{-k}) - \hat{\psi}_{\cf}(\theta)}{\hat \sigma_{\cf}(\theta)}} } \notag \\
    &= \frac{\sqrt n}{K} \sum_k \bra{P_n^k - P^k} \squ{ \xi\cur{\frac{\phi(\theta; \hat \Lambda_{-k}, \hat \mu_{-k}) - \hat{\psi}_{\cf}(\theta)}{\hat \sigma_{\cf}(\theta)}} } \label{eq:xi_ind} \\
    &= \frac{1}{\sqrt K} \sum_k  G_n^k \squ{ \xi \cur{\frac{\phi(\theta; \hat \Lambda_{-k}, \hat \mu_{-k}) - \hat{\psi}_{\cf}(\theta)}{\hat \sigma_{\cf}(\theta)}} }, \notag
\end{align}
where \eqref{eq:xi_ind} follows from the independence of the multipliers and the data together with $E(\xi)=0$, which yields $P^k \squ{\xi\cur{\phi(\theta; \hat \Lambda_{-k}, \hat \mu_{-k}) - \hat{\psi}_{\cf}(\theta)}/\hat \sigma_{\cf}(\theta) } =0$.

Define the processes
\begin{align*}
    \widetilde \Psi_n^*(\theta) &= \frac{1}{\sqrt K} \sum_k  G_n^k \squ{ \xi \cur{\frac{\phi(\theta; \hat \Lambda_{-k}, \hat \mu_{-k}) - \hat{\psi}_{\cf}(\theta)}{\sigma(\theta)}} }, \\
    \Psi_n^*(\theta) &= \frac{1}{\sqrt K} \sum_k  G_n^k \squ{ \xi \cur{\frac{\phi(\theta; \Lambda_0, \mu_0) - \psi(\theta)}{\sigma(\theta)}} } .
\end{align*}

Following the proof of Theorem 4 in \citet{kennedy2019nonparametric}, which relies on the Gaussian approximation results of \citet{chernozhukov2014gaussian}, and using the result of \thmref{cf}, it suffices to verify
\begin{equation*}
    \nsuptheta{\widehat{\Psi}_n^* - \Psi_n^* } = o_p(1)
\end{equation*}
in order to establish \thmref{ui_const}.
Similar to the argument used to obtain \eqref{eq:bridge}, under the assumptions of \thmref{ui_const}, it further suffices to show that
$$
\nsuptheta{ \widetilde{\Psi}_n^* -\Psi_n^* }  = o_p(1).
$$

We have
\begin{align}
    \nsuptheta{\widetilde{\Psi}_n^* -\Psi_n^*} =& \nsuptheta{\frac{1}{\sqrt K} \sum_k  G_n^k \squ{ \xi \cur{\frac{\phi(\theta; \hat \Lambda_{-k}, \hat \mu_{-k}) - \hat{\psi}_{\cf}(\theta)}{\sigma(\theta)} - \frac{\phi(\theta; \Lambda_0, \mu_0) - \psi(\theta)}{\sigma(\theta)} } }} \notag \\
    \leq& \nsuptheta{ \frac{1}{\sqrt K \sigma(\theta)} \sum_k  G_n^k \squ{ \xi \cur{ \phi(\theta; \hat \Lambda_{-k}, \hat \mu_{-k}) - \phi(\theta; \Lambda_0, \mu_0) } } } \label{eq:ui_df_1} \\
    &+ \nsuptheta{ \frac{1}{\sqrt K \sigma(\theta)} \sum_k  G_n^k \squ{ \xi \cur{ \hat{\psi}_{\cf}(\theta) - \psi(\theta) } } }.  \label{eq:ui_df_2}
\end{align}

For \eqref{eq:ui_df_1}, following the argument used to show $\|B_{n,1}\|_{\Theta} = o_p(1)$ in the proof of \thmref{cf}, and using the fact that $\xi$ is independent of the data with unit variance, we obtain 
$$ \nsuptheta{ \frac{1}{\sqrt K \sigma(\theta)} \sum_k  G_n^k \squ{ \xi \cur{ \phi(\theta; \hat \Lambda_{-k}, \hat \mu_{-k}) - \phi(\theta; \Lambda_0, \mu_0) } } } =o_p(1).$$

For \eqref{eq:ui_df_2}, following the proof of \thmref{eifcons}, under Assumptions~\ref{assump:bound} - \ref{assump:cf}, we can show that $\suptheta \abs{\hat{\psi}_{\cf}(\theta) - \psi(\theta)} = o_p(1)$. Then
$$
\nsuptheta{ \frac{1}{\sqrt K \sigma(\theta)} \sum_k  G_n^k \squ{ \xi \cur{ \hat{\psi}_{\cf}(\theta) - \psi(\theta) } } } \lesssim \sum_k \abs{G_n^k \xi} \cdot \suptheta \abs{\hat{\psi}_{\cf}(\theta) - \psi(\theta)} = o_p(1),
$$
since $\abs{G_n^k \xi} = O_p(1)$ for each $k$ by the central limit theorem.

Therefore, $\nsuptheta{\widetilde{\Psi}_n^* -\Psi_n^*}  = o_p(1)$. This completes the proof.
\end{proof}

\section{Proofs of Lemmas} \label{sec:pf_lem}

\subsection{Proof of \lemref{sup_tv_norm_bound}} \label{sec:pf_sup_tv_norm_bound}

\begin{proof}
The first two bounds are implied by \eqref{eq:bound_int_delam} and \eqref{eq:bound_exp_theta}, respectively.

For the third bound, using integration by parts we have
\begin{align*}
&\supt \abs{ D G(t,l;\hat\Lambda)[h] - D G(t,l;\Lambda_0)[h] }
= \supt \abs{\int_0^t \{\theta(v,l)-1\} d \cur{e^{\hat\Lambda(v|l)} - e^{\Lambda_0(v|l)}} h(v,l)} \\
&=\supt \abs{\left.\{\theta(v,l)-1\} \cur{e^{\hat\Lambda(v|l)} - e^{\Lambda_0(v|l)}} h(v,l) \right|_0^t - \int_0^t \cur{e^{\hat\Lambda(v|l)} - e^{\Lambda_0(v|l)}} h(v,l) d \{\theta(v,l)-1\}} \\
&\lesssim \supt \abs{\cur{e^{\hat\Lambda(t|l)} - e^{\Lambda_0(t|l)}} h(t,l)} \\
&\lesssim e^{\hat \Lambda(\tau|l)} \cdot \suplaml \cdot \supt \abs{h(t,l)},
\end{align*}
where  the first and second inequalities follow from \assumpref{bound} and \eqref{eq:bound_elam}, respectively.

For the last bound, we have
\begin{align*}
& \supt \abs{DS(t,l;\hat\Lambda)[h] - DS(t,l;\Lambda_0)[h]} \\
& = \supt \abs{S(t,l;\hat\Lambda) \int_0^t \theta(v,l) d h(v,l) - S(t,l;\Lambda_0) \int_0^t \theta(v,l) d h(v,l)} \\
& \leq \supt \abs{S(t,l;\hat\Lambda) - S(t,l;\Lambda_0)} \cdot \supt \abs{\int_0^t \theta(v,l) d h(v,l)} \\
& \lesssim e^{\sup_{t,l} |\theta(t,l)| \cdot \hat \Lambda(\tau|l)} \cdot \suplaml \cdot \supt \abs{h(t,l)},
\end{align*}
where the last inequality follows from \eqref{eq:bound_exp_theta}, applying integration by parts to $\int_0^t \theta(v,l) d h(v,l)$, and \assumpref{bound}.
\end{proof}

\subsection{Proof of \lemref{cond_exp}} \label{sec:pf_cond_exp}

\begin{proof}
Similar to the decomposition in \eqref{eq:pfcons_diff} used in the proof of \thmref{eifcons}, we consider
\begin{equation}
    E\cur{\left.\phi(\theta; \hat \Lambda, \hat \mu) - \phi(\theta; \Lambda_0, \mu_0) \right| O'} = D_1 + D_2 + D_3, \label{eq:D_decomposition}
\end{equation}
where
\begin{align*}
D_1 &= E \left\{ \left. \phi(\theta; \hat \Lambda, \mu_0) - \phi(\theta; \Lambda_0, \mu_0) \right| O' \right\}, \\
D_2 &= E \left\{ \left. \phi(\theta; \Lambda_0, \hat \mu) - \phi(\theta; \Lambda_0, \mu_0) \right| O' \right\}, \\
D_3 &= E \left[ \left. \left\{ \phi(\theta; \hat \Lambda, \hat \mu) - \phi(\theta; \Lambda_0, \hat \mu) \right\} - \left\{ \phi(\theta; \hat \Lambda, \mu_0) - \phi(\theta; \Lambda_0, \mu_0) \right\} \right| O' \right].
\end{align*}
We leverage the Gateaux derivative to bound each term above. 

For $D_1$, define $F(\Lambda) \coloneqq E\{\phi(\theta;\Lambda,\mu_0)\}$.  
We have
$$
D_1 = E\cur{\left. \phi(\theta; \hat \Lambda, \mu_0) - \phi(\theta; \Lambda_0, \mu_0) \right| O'} = F(\hat\Lambda) - F(\Lambda_0).
$$
Write 
\begin{equation}
\phi(\theta; \Lambda, \mu_0) 
= \phi_1(\theta; \Lambda, \mu_0) - \phi_2(\theta; \Lambda, \mu_0) + \phi_3(\theta; \Lambda, \mu_0), \label{eq:d1_phi}
\end{equation}
where
\begin{align*}
\phi_1(\theta; \Lambda, \mu_0) &= Y \theta(U,L)^{\Delta} e^{-\int_0^U\{\theta(v,L) - 1\}d\Lambda(v|L)}, \\
\phi_2(\theta; \Lambda, \mu_0) &= \int_0^\tau \mu_0(u, L) \left[\int_0^{U \wedge u} \left\{\theta(v, L)-1\right\} d e^{\Lambda(v | L)} \right] d e^{-\int_0^u \theta(v, L) d \Lambda(v | L)}, \\
\phi_3(\theta; \Lambda, \mu_0) &= \{\theta(U, L) - 1\} e^{\Lambda(U | L)} \int_{U+}^\tau \mu_0(u, L) d e^{-\int_0^u \theta(v, L) d \Lambda(v | L)}.
\end{align*}
For simplicity of notation, we suppress the dependency of $\phi, \phi_1, \phi_2$ and $\phi_3$ on $\theta$ and $\mu_0$. We can verify that the Gateaux derivative of $\phi$ at $\Lambda$ in the direction $h$ is
\begin{equation}
D\phi(\Lambda)[h] \coloneqq D\phi(\theta;\Lambda,\mu_0)[h] = D\phi_1(\Lambda)[h] - D\phi_2(\Lambda)[h] + D\phi_3(\Lambda)[h], \label{eq:G_phi_Lambda}
\end{equation}
where
\begin{align*}
D\phi_1(\Lambda)[h] &= -Y \theta(U,L)^\Delta e^{-\int_0^U\{\theta(v,L)-1\} d\Lambda(v|L)} \int_0^U\{\theta(v,L)-1\} dh(v|L), \\
D\phi_2(\Lambda)[h] &= \int_0^\tau \mu_0(u,L) \cur{DG(U \wedge u,L;\Lambda)[h] dS(u,L;\Lambda) + G(U \wedge u,L;\Lambda) d DS(u,L;\Lambda)[h]}, \\
D\phi_3(\Lambda)[h] &= \{\theta(U, L) - 1\} e^{\Lambda(U | L)} \\
&\phantom{= \{} \cdot \cur{h(U|L) \int_{U+}^\tau \mu_0(u,L) dS(u,L;\Lambda) 
+ \int_{U+}^\tau \mu_0(u,L) d DS(u,L;\Lambda)[h]}.
\end{align*}
By \eqref{eq:bound_elam}, \eqref{eq:bound_exp_theta_1}, Lemma \ref{lem:sup_tv_norm_bound}, and integration by parts, it is straightforward to show that under Assumptions~\ref{assump:bound} - \ref{assump:nuiest}, 
\begin{align}
|D\phi(\hat\Lambda)[h] - D\phi(\Lambda_0)[h]| 
\lesssim & |Y| \cdot \suplam \cdot \suph \notag \\
& + \suplam \cdot \suph, \label{eq:G_diff}
\end{align}
where the implicit constant depends only on the uniform bounds of $\theta(t,l)$, $\mu_0(t,l)$, $\Lambda_0(t,l)$, and $\hat\Lambda(t,l)$ over $t \in [0,\tau]$ and $l \in \cL$.

Substituting $\Lambda_0 + t(\hat\Lambda-\Lambda_0)$ for $\hat\Lambda$, setting $h=\hat\Lambda-\Lambda_0$,  we have
\begin{align}
    &  \abs{ D\phi(\Lambda_0 + t(\hat\Lambda-\Lambda_0))[\hat\Lambda-\Lambda_0] - D\phi(\Lambda_0)[\hat\Lambda-\Lambda_0] } \notag \\
     \lesssim & \squ{ |Y| \suplam^2+ \suplam^2} \cdot t. \label{eq:G_diff2}
\end{align}

Following Theorem 51 in \cite{vainberg1964variational}, we have 
$$\frac{d}{dt} F(\Lambda_0 + t(\hat\Lambda-\Lambda_0)) = DF(\Lambda_0 + t(\hat\Lambda-\Lambda_0))[\hat\Lambda-\Lambda_0],$$
$\forall t\in[0,1]$. 
Integrating both sides with respect to $t$ gives
\begin{align*}
F(\hat\Lambda) - F(\Lambda_0) 
=& \int_0^1 \frac{d}{dt} F(\Lambda_0 + t(\hat\Lambda-\Lambda_0)) dt 
= \int_0^1 DF(\Lambda_0 + t(\hat\Lambda-\Lambda_0))[\hat\Lambda-\Lambda_0] dt \\
=& \int_0^1 \cur{DF(\Lambda_0 + t(\hat\Lambda-\Lambda_0))[\hat\Lambda-\Lambda_0] - DF(\Lambda_0)[\hat\Lambda-\Lambda_0]} dt,
\end{align*}
where the last equality follows because $ DF(\Lambda_0)[\hat\Lambda-\Lambda_0]=\left.\left. {\partial} E\cur{\phi\bra{\theta; \Lambda_0 + t(\hat\Lambda - \Lambda_0)}} \right/ {\partial t} \right|_{t=0} = 0$ due to Neyman orthogonality  
of the efficient influence function. 

Here we show that the order of Gateaux differentiation and the expectation in $F$ can be exchanged. 
Under Assumptions~\ref{assump:bound} - \ref{assump:nuiest},  using arguments similar to those in the proof of \thmref{eifcons}, we can verify that for $t \in \bbR$,
$$
\abs{\frac{1}{t} \cur{ \phi(\Lambda_0 + t(\hat\Lambda - \Lambda_0))-\phi(\Lambda_0)} } \lesssim (|Y|+1) \suplam,
$$
where the RHS is integrable. By dominated convergence theorem, 
\begin{align*}
    D F(\Lambda_0)[\hat\Lambda - \Lambda_0] 
    &= \lim_{t \to 0} \frac{1}{t} \cur{F(\Lambda_0 + t (\hat\Lambda - \Lambda_0)) - F(\Lambda_0) } \\
    &= \lim_{t \to 0} E \squ{ \frac{1}{t} \cur{ \phi(\Lambda_0 + t (\hat\Lambda - \Lambda_0)) - \phi(\Lambda_0)} } \\
    &= E \squ{ \lim_{t \to 0} \frac{1}{t} \cur{ \phi(\Lambda_0 + t (\hat\Lambda - \Lambda_0)) - \phi(\Lambda_0)} } 
    = E \cur{ D\phi(\Lambda_0)[\hat\Lambda - \Lambda_0]}.
\end{align*}
Likewise, $DF(\Lambda_0 + t(\hat\Lambda-\Lambda_0))[\hat\Lambda-\Lambda_0] = E\cur{ D\phi(\Lambda_0 + t(\hat\Lambda-\Lambda_0))[\hat\Lambda-\Lambda_0] }$. 

Thus,
\begin{align}
    & \abs{DF(\Lambda_0 + t(\hat\Lambda-\Lambda_0))[\hat\Lambda-\Lambda_0] - DF(\Lambda_0)[\hat\Lambda-\Lambda_0]} \notag \\
    & = \abs{E\cur{D\phi(\Lambda_0 + t(\hat\Lambda-\Lambda_0))[\hat\Lambda-\Lambda_0] - D\phi(\Lambda_0)[\hat\Lambda-\Lambda_0]}} \notag \\
    & \lesssim E \squ{ |Y| \suplam^2+ \suplam^2} \cdot t, \label{eq:DF_bound}
\end{align}
where \eqref{eq:DF_bound} follows from \eqref{eq:G_diff2}.
Therefore,
\begin{align}
|D_1| &= |F(\hat\Lambda) - F(\Lambda_0)| \notag\\
&\leq \int_0^1 \abs{DF(\Lambda_0 + t(\hat\Lambda-\Lambda_0))[\hat\Lambda-\Lambda_0] - DF(\Lambda_0)[\hat\Lambda-\Lambda_0]} dt \notag \\
&\lesssim \int_0^1 E\squ{|Y| \suplam^2} t dt
+ \int_0^1 E\squ{\suplam^2} t dt \label{eq:D2_condition} \\
& \lesssim E\bra{Y^2}^{\frac{1}{2}} \cdot E\squ{\suplam^4}^{\frac{1}{2}}
+ E\squ{\suplam^2} \label{eq:cs_sum} \\
&\lesssim \|\hat\Lambda-\Lambda_0\|_{\dagger,\sup,4}^2 + \dgsupnlam^2, \label{eq:d1_bound}
\end{align}
where \eqref{eq:D2_condition}, \eqref{eq:cs_sum}, and \eqref{eq:d1_bound} follow from \eqref{eq:DF_bound}, the Cauchy–Schwarz inequality, and $E\bra{Y^2} < \infty$ under \assumpref{bound}, respectively.
If $Y$ is bounded, then 
$$E\cur{|Y| \suplam^2} \lesssim E\cur{\suplam^2}$$ 
in \eqref{eq:D2_condition}; consequently, the bound in \eqref{eq:d1_bound} reduces to
\begin{equation}
    |D_1|\lesssim \dgsupnlam^2 \label{eq:d1_bound_relaxed}.
\end{equation}

For $D_2$, we apply the same Gateaux derivative approach to establish that
\begin{equation}
    D_2 = E\cur{ \left.\phi(\theta; \Lambda_0, \hat \mu) - \phi(\theta; \Lambda_0, \mu_0) \right|O'}= 0 \label{eq:d2}.
\end{equation}
More specifically, since $\phi(\theta; \Lambda_0, \mu)$ is affine in $\mu$, we can show that its Gateaux derivatives at different $\mu$ values are identical, i.e.~$D\phi(\mu)[h] = D\phi(\mu')[h]$, for any $\mu,\mu'$ and direction $h$.
Define $H(\mu) \coloneqq E \cur{\phi(\theta;\Lambda_0,\mu)}$, where we suppress the dependency of $\phi$ on $\theta$ and $\Lambda_0$. Then, similar to the argument leading to \eqref{eq:d1_bound},  
\begin{align*}
    |D_2| &= |H(\hat\mu) - H(\mu_0)| \\
    &\leq \int_0^1 \abs{DH(\mu_0 + t(\hat\mu-\mu_0))[\hat\mu-\mu_0] - DH(\mu_0)[\hat\mu-\mu_0]} dt \\
    &\leq \int_0^1 \abs{E\cur{D\phi(\mu_0 + t(\hat\mu-\mu_0))[\hat\mu-\mu_0] - D\phi(\mu_0)[\hat\mu-\mu_0]}} dt = 0.
\end{align*}
This concludes \eqref{eq:d2}. 

For $D_3$, since $\phi(\theta; \cdot, \mu)$ is affine in $\mu$, this term has a structure similar to $D_1$: 
\begin{align*}
    D_3 &= E \left[ \left. \left\{ \phi(\theta; \hat \Lambda, \hat \mu) - \phi(\theta; \Lambda_0, \hat \mu) \right\} - \left\{ \phi(\theta; \hat \Lambda, \mu_0) - \phi(\theta; \Lambda_0, \mu_0) \right\} \right| O' \right] \\
    &= E\cur{\left. -\phi_2(\theta; \hat \Lambda, \hat\mu - \mu_0) + \phi_3(\theta; \Lambda_0, \hat\mu - \mu_0) \right| O'},
\end{align*}
where $\phi_2$ and $\phi_3$ are defined in \eqref{eq:d1_phi}.
By applying the same bounding procedure as for $D_1$, we omit the analogous details and obtain
\begin{equation}
    D_3 \lesssim \dgsupnmu \cdot \dgsupnlam. \label{eq:d3_bound}
\end{equation}
Intuitively, $D_1$ depends only on the estimation error of $\hat\Lambda$, while $D_3$ additionally involves the estimation error of $\hat\mu$.

Combining \eqref{eq:D_decomposition} with \eqref{eq:d1_bound}, \eqref{eq:d2}, and \eqref{eq:d3_bound}, and applying the triangle inequality, yields
\begin{align}
& \abs{E\cur{ \left.\phi(\theta; \hat \Lambda, \hat \mu) - \phi(\theta; \Lambda_0, \mu_0) \right| O'}} \notag \\
& \lesssim \dgsupnmu \cdot \dgsupnlam + \dgsupnlam^2
+ \|\hat\Lambda-\Lambda_0\|_{\dagger,\sup,4}^2, \label{eq:cf_const_remainder_rate}
\end{align}
where the implicit constant 
in the upper bound depends on $\theta$ only through $\sup_{t,l} \theta(t,l)$. 
If $Y$ is bounded, by \eqref{eq:d1_bound_relaxed},
the bound reduces to $ \dgsupnmu \cdot \dgsupnlam + \dgsupnlam^2$. 
\end{proof}


\end{document}